\documentclass[a4paper, intlimits, 10pt]{amsart}

\usepackage[T1]{fontenc}
\usepackage{float}

\usepackage{times}
\usepackage[english]{babel}
\usepackage[T1]{fontenc}
\usepackage{enumitem}
\usepackage[short]{optidef}

\usepackage{graphicx}
\usepackage{amsmath,bm}
\usepackage{caption}
\usepackage{subcaption}
\usepackage{booktabs}
\usepackage{tabularx}
\usepackage{color}
\usepackage{eurosym}
\usepackage[numbers,sort&compress]{natbib}

\usepackage{amsthm, amssymb, amsfonts}
\usepackage{amsmath}

\DeclareMathOperator*{\argmin}{arg\,min}

\usepackage{accents}

\newtheorem{theorem}{Theorem}[section]

\newtheorem{proposition}[theorem]{Proposition}

\newtheorem{example}[theorem]{Example}

\newtheorem{remark}[theorem]{Remark}
\numberwithin{equation}{section}

\begin{document}

\title{Relative Bound and Asymptotic Comparison of Expectile with Respect to Expected Shortfall}

\date{\today}

\author{Samuel Drapeau}
\thanks{National Science Foundation of China, Grants Numbers: 11971310 and 11671257; Grant ``Assessment of Risk and Uncertainty in Finance'' number AF0710020 from Shanghai Jiao Tong University; are gratefully acknowledged.}
\address{School of Mathematical Sciences \& Shanghai Advanced Institute for Finance (CAFR)\newline Shanghai Jiao Tong University, Shanghai, China}
\email{sdrapeau@saif.sjtu.edu.cn}
\urladdr{http://www.samuel-drapeau.info}

\author{Mekonnen Tadese}
\address{School of Mathematical Sciences\newline Shanghai Jiao Tong University, Shanghai, China}
\email{mekonnenta@sjtu.edu.cn}
\thanks{Both authors thanks St\'ephane Cr\'epey and Hans F\"ollmer for fruitful discussions.}

\begin{abstract}
    Expectile bears some interesting properties in comparison to the industry wide expected shortfall in terms of assessment of tail risk.
    We study the relationship between expectile and expected shortfall using duality results and the link to optimized certainty equivalent.
    Lower and upper bounds of expectile are derived in terms of expected shortfall as well as a characterization of expectile in terms of expected shortfall.
    Further, we study the asymptotic behavior of expectile with respect to expected shortfall as the confidence level goes to $1$ in terms of extreme value distributions.
    We use concentration inequalities to illustrate that the estimation of value at risk requires larger sample size than expected shortfall and expectile for heavy tail distributions when $\alpha$ is close to $1$.
    Illustrating the formulation of expectile in terms of expected shortfall, we also provide explicit or semi-explicit expressions of expectile and some simulation results for some classical distributions.

    \vspace{5pt}

    \noindent
    {Keywords:} Expectile; Expected Shortfall; Value at Risk; Extreme Value; Risk Measure. 
\end{abstract}

\maketitle
\section{Introduction}
The expectile is a generalization of quantile introduced by \citet{newey1987}.
It is defined as the argmin of a quadratic loss
\begin{equation*}
    e_\alpha(L)=\argmin\left\{ \alpha E\left[ \left( (L-m)^+ \right)^2 \right]+(1-\alpha)E\left[ \left( (L-m)^- \right)^2 \right] \right\}.
\end{equation*}
For $1/2\leq \alpha< 1$, the expectile is a coherent risk measure that corresponds to \citet{foellmer2002}'s shortfall risk with loss function $\ell(x)=\alpha x^+-(1-\alpha) x^-$.
Widely used in insurance and statistics, it has recently gained some interest in finance as it bears some interesting features for the assessment of tail risk in comparison to the industry wide expected shortfall risk measure introduced by \citet{artzner1999}.
From its definition, expectile is elicitable, which is a useful property in terms of backtesting, see \citet{gneiting2011}, \citet{bellini2015}, \citet{emmer2015}, \citet{ziegel2016}, and \citet{chen2018} for a discussion about the financial relevance.
In the seminal paper \citet{weber2006}, and later \citet{bellini2015}, \citet{ziegel2016}, \citet{delbaen2016}, it actually turns out that expectile is the only elicitable risk measure within the class of coherent and law invariant risk measures.
Expectile is also invariant under randomization, while expected shortfall is not, see \citet{weber2006} and \citet{guo2019}.
The property of invariance under randomization is closely related to the convexity of the acceptance set and rejection set of a risk measure, when risk is defined on the space of distributions.
That is, for expectile if both $L_1$ and $L_2$ are acceptable, then the randomized position  $L=L_1$ with probability $p$ and $L=L_2$ with probability $1-p$ with $p$ in $[0,1]$ is also acceptable, see \citep{weber2006} for the detail.
Finally, multivariate shortfall risk -- expectile being an example of which -- seems to be suitable in terms of systemic risk management and risk allocation, see \citet{armenti2018}.
Due to these appealing properties, several authors suggest expectile as an alternative to expected shortfall and value at risk, see \citep{emmer2015, bellini2014, bellini2017, bellini2015, chen2018} for instance.

The goal of this paper is to study the relationship between expectile and expected shortfall.
More specifically, the objective is to provide lower and upper bounds of expectile in terms of expected shortfall, formulate explicitly expectile and its Euler allocation as a function of expected shortfall, and compare the asymptotic behavior of expectile and its Euler allocation with respect to expected shortfall as the confidence level goes to $1$.
As for the bounds, our approach is based on duality results and the link between expectile and expected shortfall through optimized certainty equivalent.
For loss profile $L$ with zero mean, our first result mainly focus on the bounds 
\begin{equation}\label{eq:optbound}
    \left( 1-\frac{1-\alpha}{\alpha + (1-2\alpha)\beta} \right)ES_{\beta}(L)\leq
    e_{\alpha}(L)\leq \left( 1-\frac{1-\alpha}{\alpha} \right)ES_{\alpha}(L).
\end{equation}
As shown in Proposition \ref{prop:OGALA}, the optimal lower bound is in fact an equality
\begin{equation*}
    e_{\alpha}(L)
    =\left(1-\frac{1-\alpha}{\alpha + (1-2\alpha)\beta^\ast}\right)ES_{\beta^\ast}(L)
\end{equation*}
where $\beta^\ast \in [P\left[ L< e_{\alpha}(L) \right], P\left[ L\leq  e_{\alpha}(L) \right]] $.
For continuous distribution, the expression of $\beta^\ast$ is mentioned in \citet[Equation 7]{tylor2008} based on results by \citet{newey1987}.
We generalized this result to any distribution using optimized certainty equivalent.
As an application of this relation we can easily derive explicit or semi-explicit formulations of expectile for wide classes of distributions.
Let $L_1,\dots,L_d$ be loss profiles such that $L=\sum_{k=1}^d L_k$. 
Under some smoothness assumptions, in the same sprits of the optimal lower bound, the Euler allocations of expectile can also be formulated as a function of the Euler allocation of expected shortfall given by  
\begin{equation*}
    e_\alpha(L_k|L)=\left( 1-\frac{1-\alpha}{\alpha+ (1-2\alpha)\beta^\ast} \right)ES_{\beta^\ast}(L_k|L)+\frac{1-\alpha}{\alpha +(1-2\alpha)\beta^\ast}E[L_k].
\end{equation*}
As for the upper bound, \citet{delbaen2013} and \citet{ziegel2016} provide a comonotone least upper bound of expectile in terms of concave distortion risk measure.
Using this result, we show that the upper bound given by Relation \eqref{eq:optbound} is the smallest within the class of expected shortfalls dominating expectile.

According to these bounds, expected shortfall is more conservative than expectile.
We therefore, address their comparative asymptotic behavior as the confidence level goes to $1$.
In actuarial literature, asymptotic analysis is a subject of intensive research as it helps risk managers to model large losses with small amounts of data and to establish asymptotic relationships between risk measures, see \citet{hua2011}.
While \citet{hua2011}, \citet{tang2012} and \citet{mao2012} establish asymptotic relationship between expected shortfall and value at risk, \citet{bellini2017} and \citet{mao2015} provides asymptotic analysis of expectile in terms of value at risk when the loss profile belongs to the maximum domain of attraction of extreme value distributions.
There is in particular an asymptotic relationship between value at risk and expectile excerpting the tail index of the distribution belonging to the Frechet type for instance.
It is therefore possible to estimate the tail index of a distribution by comparing the empirical value of value at risk and expectile for large $\alpha$.
However, from an asymptotic point of view the estimation of value at risk may need large sample size as compared to expected shortfall and expectile for heavy tailed distributions.
From recent concentration inequality results from \citet{fournier2015} using Wasserstein distance, we provide estimations of the error for the empirical expected shortfall $ES_{\alpha,n}$ and empirical expectile $e_{\alpha,n}$ as a function of the confidence level $\alpha$ and their corresponding required sample size $n_{ES_{\alpha}}$ and $n_{e_{\alpha}}$.
For instance, in the case where the distribution has some moment $q>2$, we obtain
\begin{align*}
    P\left[ \left| ES_{\alpha, n} -ES_{\alpha}\right|\geq \varepsilon \right] & \leq C_1 n^{1-s}\varepsilon^{2(1-s)}(1-\alpha)^{2(1-s)}\\
    P\left[ \left| e_{\alpha, n} -e_{\alpha}\right|\geq \varepsilon \right] & \leq C_2 n^{1-s}\varepsilon^{2(1-s)}\left(\frac{1-\alpha}{\alpha}\right)^{2(1-s)}
\end{align*}
for any $2<s<q$ where the constant $C_1$ and $C_2$ are independent of $n$, $\varepsilon$ and $\alpha$.
In particular, as showed in Proposition \ref{prop:samplesize}, for Fr\'echet type distributions with a moment $q>2$, both $n_{ES_{\alpha}}$ and $n_{e_{\alpha}}$ are of the order $1/(1-\alpha)^2$ which are infinitesimal with respect to the corresponding $n_{q_\alpha}$ for the quantile as $\alpha$ goes to one.
Note that estimation for the empirical estimation of the expected shortfall and expectiles and more general risk measures has been the subject of recent studies, see \citet{gao2011}, \citet{holzmann2016}, \citet{kolla2019}, and \citet{daniel2020} for instance.
We use here the bounds in \citep{fournier2015} to get the explicit dependence in terms of the confidence level $\alpha$ which is new to our knowledge.

Using related results, when the loss profile belongs to the domain of attraction of either Weibull type $MDA(\Psi_\eta)$, Gumbel type $MDA(\Lambda)$ or Fr\'{e}chet type $MDA(\Phi_\eta)$, we establish asymptotic relationship between expectile and expected shortfall by providing both the first-order and second-order asymptotic expansion.
For a Fr\'echet type tail distribution with $\eta>1$, asymptotically the ratio of expectile to expected shortfall become strictly less than $1$.
In this case, it actually hold
\begin{equation*}
    e_{\alpha}(L) \sim \frac{(\eta - 1)^{\frac{\eta -1}{\eta}}}{\eta} ES_{\alpha}(L) \quad \text{and} \quad e_{\alpha}(L_k|L) \sim \frac{(\eta - 1)^{\frac{\eta -1}{\eta}}}{\eta} ES_{\alpha}(L_k|L).
\end{equation*}
This result also show that the upper bound provided by Relation \eqref{eq:optbound} is not asymptotically equivalent to $e_\alpha(L)$ in general.
It also allows an estimation of tail index based on empirical data.

We also consider the asymptotic behavior of the parameter $\beta^\ast$. 
For loss profiles whose distribution belongs to Fr\'{e}chet type $MDA(\Phi_\eta)$ with $\eta>1$, \citet{bellini2014} provide the asymptotic behavior of $\beta^\ast$ in terms of $\alpha$.
For Weibull type $MDA(\Psi_\eta)$ and Gumbel type $MDA(\Lambda)$, we show that $1-\alpha=o(1-\beta^\ast)$.
For Fr\'{e}chet case, we also provide a second-order asymptotic expansion for $(1-\beta^\ast)/(1-\alpha)$.

The paper is organized as follows.
In Section 2, aside definitions and notations, we revisit the link between expectile and expected shortfall through optimized certainty equivalent.
In Section 3, we address the lower and upper bounds of expectile in terms of expected shortfall as well as characterize expectile and its Euler allocations in terms of expected shortfall.
Section 4 focuses on asymptotic behavior of expectile in terms of expected shortfall according to the maximum domain of attractions of extreme value distributions to which the loss profile belongs.
Section 5 illustrate the results of Section 3 in terms of explicit or semi-explicit expression of expectile for commonly known distributions.
It also provide an illustrations for some of the asymptotic results of Section 4. 

\section{ Expectile Versus Expected Shortfall through Optimized Certainty Equivalent}
Let $(\Omega, \mathcal{F},P)$ be a probability space and $L^1$ be the set of integrable random variables identified in the almost sure sense.
For $a>0$ and $b\geq 0$ with $1/a\geq b$, denote by 
\begin{equation*}
    \mathcal{Q}_{a,b}=\left\{Q\ll P\colon b\leq \frac{dQ}{dP}\leq \frac{1}{a}\right\}.
\end{equation*}
Throughout, elements of $L^1$ are generically denoted by $L$ and considered as a loss profile.
Given such an $L$ in $L^1$, we denote by $F_L$ and $q_L$ its cumulative distribution and left-quantile function, respectively, that is
\begin{equation*}
    q_L(u) = \inf \left\{ m \colon F_L(m):=P[L\leq m] \geq u \right\}.
\end{equation*}
We also denote the right-quantile function of $L$ by $q^{+}_L$, that is $q^{+}_L(u) = \inf \{ m \colon F_L(m)> u \}$.
A function $R\colon L^1 \to \mathbb{R}$ is called a \emph{risk measure} if it is
\begin{enumerate}[label=\textbf{(\Roman*)}]
    \item \textbf{quasi-convex:} $R(\lambda L_1+(1-\lambda)L_2)\leq \max\{R(L_1), R(L_2)\}$ for every $0\leq \lambda \leq 1$.
    \item \textbf{monotone:} $R(L_1)\leq R(L_2)$ whenever $L_1\leq L_2$ almost surely.
\end{enumerate}
A risk measure is further called \emph{monetary} if it is additionally
\begin{enumerate}[label=\textbf{(\Roman*)}, resume]
    \item \textbf{cash-invariant:} $R(L-m)=R(L)-m$ for every $m$ in $\mathbb{R}$.
\end{enumerate}
Finally, a monetary risk measure is called \emph{coherent} if it is additionally
\begin{enumerate}[label=\textbf{(\Roman*)}, resume]
    \item \textbf{sub-additive:} $R(L_1+L_2)\leq R(L_1)+R(L_2)$.
\end{enumerate}
It is known that monetary risk measures are automatically convex, and, coherent monetary risk measures are \emph{positive-homogeneous}\footnote{$R(\lambda L)=\lambda R(L)$ for every $\lambda>0$.}.
For $L$ in $L^1$, we define
\begin{enumerate}[label=$\bullet$, fullwidth]
    \item \emph{Value at Risk}: for $0<\alpha <1$,
        \begin{equation*}
            V@R_\alpha(L)=\inf\left\{ m\colon P\left[ L\leq m \right]\geq \alpha \right\}=q_L(\alpha).
        \end{equation*}
    \item \emph{Expected Shortfall}: for $0\leq  \alpha < 1$,
        \begin{equation*}
            ES_{\alpha}(L)=\frac{1}{1-\alpha}\int_\alpha^1 V@R_u(L) du = \frac{1}{1-\alpha}\int_{\alpha}^{1}q_L(u)du.
        \end{equation*}
    \item \emph{Expectile}: for $1/2\leq \alpha < 1$, the $\alpha$-expectile of $L$ is defined as a unique number $e_\alpha(L)$  solving 
        \begin{equation*}
            \alpha E\left[ (L-e_\alpha(L))^+ \right]= (1-\alpha) E\left[ (L-e_\alpha(L))^- \right].
        \end{equation*}
\end{enumerate}

The value at risk is cash invariant, monotone and positive-homogeneous, it is however not sub-additive, see \citep{artzner1999,tasche2002}.
The expected shortfall is a special case of an optimized certainty equivalent, while the expectile corresponds to the shortfall risk with loss function $\ell(x)=\alpha x^+-(1-\alpha) x^-$ in the standard definition, \citep{weber2006,ziegel2016,foellmer2002}.
Indeed, $\ell$ is increasing, convex whenever $\alpha\geq 1/2$ and such that $\inf \ell(x)<0$ whenever $\alpha>0$.
Hence, the expectile can be seen as a scaled version of an optimized certainty equivalent, see \citep{bental2007}.
In the literature, see for instance \citep{bellini2014,newey1987}, expectile is also defined as
\begin{equation*}
    \argmin\left\{ \alpha E\left[ \left( (L-m)^+ \right)^2 \right]+(1-\alpha) E\left[ \left( (L-m)^- \right)^2 \right] \right\},
\end{equation*}
for $L$ in $L^2$.
However, due to the first order condition this coincides with the present definition.

Let us recall the following known properties of expectile and expected shortfall.
\begin{proposition}\label{prop:SRES}
    The expectile and expected shortfall are law invariant monetary risk measures and it holds
    \begin{align*}
        ES_{\alpha}(L) & = \min\left\{ m+\frac{1}{1-\alpha}E\left[ (L-m)^+ \right]\colon m\in \mathbb{R} \right\}\\
                       & = q_L(\alpha) +\frac{1}{1-\alpha}E\left[ \left( L-q_L(\alpha) \right)^+ \right]\\
                       & = \max\left\{ E^Q[L]\colon Q\in \mathcal{Q}_{1-\alpha,0} \right\}
    \end{align*}
    with optimal density 
    \begin{equation*}
        \frac{dQ^\ast}{dP}=\frac{1}{1-\alpha}\left(1_{\{L>q_L(\alpha)\}}+k1_{\{L= q_L(\alpha)\}}\right)
    \end{equation*}
    where $k$ is a constant such that $E[dQ^\ast/dP]=1$ and
    \begin{align}
        e_{\alpha}(L) & = \max_{\frac{1-\alpha}{\alpha}< \gamma <1}\left\{(1-\gamma)ES_{\frac{(1+\gamma)\alpha-1)}{(2\alpha-1)\gamma}}(L)+\gamma E[L]\right\}\label{eq:SRES01}\\
                      & = \max_{\frac{1-\alpha}{\alpha}< \gamma< 1}\int_{0}^{1} ES_{u}(L)\mu^{\gamma}(du) \label{eq:SRES02}\\
                      & = \max \left\{ E^Q[L] \colon Q\in \mathcal{Q}_{(1-\alpha)/\gamma\alpha,\gamma} \text{ for some }\gamma\in \left[\frac{1-\alpha}{\alpha},1\right]\right\}\label{eq:SRES03}
    \end{align}
    with optimal density
    \begin{equation*}
        \frac{dQ^\ast}{dP}
        =\frac{\alpha1_{\{L>e_\alpha(L)\}}+(1-\alpha) 1_{\{L \leq e_\alpha(L)\}}}{\alpha + (1- 2\alpha)P[L\leq e_\alpha(L)]}
    \end{equation*}
    where $\mu^\gamma=( 1-\gamma)\delta_{\frac{(1+\gamma)\alpha-1)}{(2\alpha-1)\gamma}}+\gamma \delta_0$ is a parameterized family of distribution on $[0,1]$.
\end{proposition}
These results can be found or derived from \citep{bental2007, foellmer2002, artzner1999, bellini2014}.
Interestingly though, they are strongly connected through the optimized certainty equivalent from \citet{bental2007}.
For the sake of readability and further computation we expose briefly this connection.
\begin{proposition}\label{prop:OCE}
    For a loss function $\ell_{a,b}(x) := x^+ / a - b x^-$ where $0<a<1$ and $0\leq b\leq 1$, the optimized certainty equivalent defined as
    \begin{equation}\label{eq:OCE01}
        R_{a,b}(L)=\inf \left\{ m+E\left[ \ell_{a,b}(L-m) \right] \colon m \in \mathbb{R}\right\}, \quad L\in L^1
    \end{equation}
    is a law invariant coherent risk measure such that
    \begin{align}
        \label{eq:OCE02} R_{a,b}(L) & = q_{L}\left( \lambda(a,b) \right)+E\left[ \ell_{a,b}\left( L-q_{L}\left(\lambda(a,b) \right) \right) \right]\\
        \label{eq:OCE03}            & = \frac{1}{a} \int_{\lambda(a,b)}^{1} q_{L}(u) du+b\int_{0}^{\lambda(a,b)} q_{L}(u) du\\
        \label{eq:OCE04}            & = (1-b)ES_{\lambda(a,b)}(L)+b E[L]\\
        \label{eq:OCE05}            & = \sup\left\{ E^Q[L]\colon Q\in \mathcal{Q}_{a,b}\right\}
    \end{align}
    where $\lambda(a,b)=(1- a)/(1 - ab)$.
    Furthermore, for $0<b\leq 1$ it holds
    \begin{multline}\label{eq:OCE06}
        \inf\left\{ m \colon E\left[ \ell_{a,b}(L-m) \right] \leq 0\right\} = \sup_{a \leq \gamma \leq 1/b}R_{a/\gamma,b\gamma}(L)\\
        = \sup\left\{ E^Q[L]\colon Q\in \mathcal{Q}_{a/\gamma,b\gamma} \text{ for some } \gamma\in [a,1/b]\right\}.
    \end{multline}
\end{proposition}
\begin{proof}
    Following \citep{bental2007}, the  optimal $m^\ast$ in definition \eqref{eq:OCE01} satisfies
    \begin{equation*}
        \frac{1}{a}P\left[ L > m^\ast \right]+bP\left[ L\leq m^\ast \right]\leq 1\leq\frac{1}{a}P\left[ L \geq m^\ast \right]+bP\left[ L< m^\ast \right].
    \end{equation*}
    Rearranging, we get $P[ L < m^\ast ]\leq \lambda(a,b)\leq P[ L\leq m^\ast ]$ showing that $m^\ast=q_L(\lambda(a,b))$.
    Plugging the optimizer into \eqref{eq:OCE01} yields \eqref{eq:OCE02}.
    From \eqref{eq:OCE02} to \eqref{eq:OCE03} comes from the fact that $q_L\sim L$.
    As for \eqref{eq:OCE04}
    \begin{equation*}
        \begin{split}
            \frac{1}{a}\int_{\lambda(a.b)}^{1}q_L(u)du + b \int_{0}^{\lambda(a,b)}q_L(u)du
            &=\left(\frac{1}{a}-b\right)\int_{\lambda(a,b)}^{1}q_L(u)du+b E[L]\\
            &=(1-b)ES_{\lambda(a,b)}(L)+bE[L].
        \end{split}
    \end{equation*}
    The Relation \eqref{eq:OCE03} implies that the optimized certainty equivalent is a law invariant and coherent risk measure.
    The Relation \eqref{eq:OCE05} follows from the general robust representation of optimized certainty equivalent in terms of divergences, that is
    \begin{equation*}
        \inf\left\{ m+E\left[ \ell_{a,b}\left( L-m \right) \right] \colon m \in \mathbb{R}\right\}=\sup\left\{ E^Q\left[ L \right] -E\left[ \ell_{a,b}^\ast\left( \frac{dQ}{dP} \right) \right]\colon \frac{dQ}{dP}\in L^\infty\right\}
    \end{equation*}
    see \citep[Theorem 4.2]{bental2007}, since the convex conjugate\footnote{ $\ell_{a,b}^\ast(x)=\sup\{x\cdot y-\ell_{a,b}(y)\colon y\in \mathbb{R}^d\}$.} $\ell_{a,b}^\ast(x) = 0$ if $b\leq x\leq 1/a$ and $\infty$ otherwise.
    As for the last Relation \eqref{eq:OCE06}, it comes from the general relation between optimized certainty equivalent and shortfall risk \citep[Section 5.2]{bental2007} where
    \begin{align*}
        \inf \left\{ m\colon E\left[ \ell_{a,b}(L-m)\leq 0 \right] \right\}  &=\sup_{1/\gamma \in \mathrm{dom}(\ell_{a,b}^\ast)}\inf\left\{ m+\gamma E\left[ \ell_{a,b}\left( L-m \right) \right] \right\}
    \end{align*}
    which gives the result.
\end{proof}
\begin{proof}[Proof of Proposition \ref{prop:SRES}]
    The relations for the expected shortfall follows directly from Proposition \ref{prop:OCE} by noticing that $ES_{\alpha}(L)=R_{a,b}(L)$ for $a = 1-\alpha$ and $b=0$.
    As for the relations for the expectile, they follow from \eqref{eq:OCE06} as $e_{\alpha}(L)=\inf\{ m\colon E[\ell_{a,b}( L-m )]\leq 0 \}$ for $a = 1/2\alpha$ and $b=2(1-\alpha)$ which fulfills the conditions of Proposition \ref{prop:OCE} as $1/2\leq \alpha<1$.
    As for the optimal density for expected shortfall, see \citet{foellmer2016, neil2015} and for expectile it is given in \citep[Proposition 8]{bellini2014}.
\end{proof}

\begin{remark}
    Relations \eqref{eq:SRES01}--\eqref{eq:SRES03} provide the link between expectile and expected shortfall.
    One sees in particular, that while expected shortfall is comonotone, the expectile is not.
    Indeed, Relation \eqref{eq:SRES02} is the Kusuoka representation which can not fulfill the assumptions of \citep[Theorem 4.93, p. 260]{foellmer2016}.
    On the other hand, as showed in \citep{weber2006} while expectile is invariant under randomization, the expected shortfall is not. 
\end{remark}

\section{Expectile as a Function of Expected Shortfall}
Based on Relation \eqref{eq:SRES01} we provide bounds for the expectile in terms of expected shortfall in the spirit of \citep[Proposition 9]{bellini2014}.
The upper bound $(1-(1-\alpha)/\alpha)ES_\alpha$ is to our knowledge new, while the larger upper bound $ES_{\frac{2\alpha-1}{\alpha}}$ is given in \citep{delbaen2013}.
The present proof uses the relation between optimized certainty equivalent and expectile.

\begin{proposition}\label{prop:bound}
    Let $L$ be in $L^1$ with zero mean\footnote{
        Due to translation invariance, in the case where $E[L]\neq 0$, we get
        \begin{multline*}
            \left( 1-\frac{1-\alpha}{\alpha + (1-2\alpha)\beta} \right)ES_{\beta}(L)+\frac{1-\alpha}{\alpha + (1-2\alpha)\beta}E[L]
            \leq \\
            e_{\alpha}(L)\leq \left( 1-\frac{1-\alpha}{\alpha} \right)ES_{\alpha}(L)+\frac{1-\alpha}{\alpha} E[L]\leq ES_{\frac{2\alpha-1}{\alpha}}(L).
        \end{multline*}
    }. 
    Then for each $0<\beta<1$, it holds
    \begin{equation*}
        \left( 1-\frac{1-\alpha}{\alpha + (1-2\alpha)\beta} \right)ES_{\beta}(L)\leq
        e_{\alpha}(L)
        \leq \left( 1-\frac{1-\alpha}{\alpha} \right)ES_{\alpha}(L)\leq ES_{\frac{2\alpha-1}{\alpha}}(L).
    \end{equation*}
\end{proposition}

\begin{proof}
    Let $1/2\leq \alpha<1$ be given.
    On the one hand, from the proof of Propositions \ref{prop:SRES} and \ref{prop:OCE}, $e_{\alpha}(L)=\sup_{ab \leq \gamma \leq 1}R_{ab/\gamma , \gamma}(L)$ for $a = 1/2\alpha$ and $b=2(1-\alpha)$.
    It implies that $e_{\alpha}(L)\geq R_{ab/\gamma,\gamma}(L)$ and therefore from \eqref{eq:OCE04} it follows that $e_{\alpha}(L)\geq (1-\gamma)ES_{\lambda(ab/\gamma,\gamma)}(L)$ for every $ab\leq \gamma\leq 1$.
    Solving $\gamma$ for $\lambda(ab/\gamma,\gamma)=\beta$ yields the left hand inequality.
    On the other hand, we have $\ell_{ab/\gamma,\gamma} \leq \ell_{ab,ab}\leq \ell_{ab,0}$, showing together with $E[L]=0$ that $e_{\alpha}(L)\leq R_{ab,ab}(L)\leq R_{ab,0}(L)$.
    Since $ab = (1-\alpha)/\alpha$ and $\lambda(ab,ab)=\alpha$, as a result of \eqref{eq:OCE04} the right hand side inequalities also hold.
\end{proof}

If we set $\beta=\alpha$, the lower bound corresponds to the one stated in \citep[Proposition 9]{bellini2014}, that is
\begin{equation}\label{eq:bound}
    \left( 1-\frac{1}{2\alpha} \right)ES_{\alpha}(L)\leq e_{\alpha}(L).
\end{equation}
As for the lower bound, from \eqref{eq:SRES01}, it is immediate that there exists $\beta^\ast$ satisfying the equality in the above proposition.
When $F_L$ is continuous, from \citep[Equation $7$]{tylor2008} we get an optimal $\beta^\ast=P[L\leq e_\alpha(L)]$. 
We generalized this result for any distribution and formulate expectile as a convex combination of expected shortfalls.

\begin{proposition}\label{prop:OGALA}
    Let $L$ be in $L^1$ not identically constant,\footnote{If $L$ is identically constant, then $\beta^\ast$ is any number in $(0,1)$.} it holds that 
    \begin{equation*}
        e_{\alpha}(L)
        =\left(1-\frac{1-\alpha}{\alpha +(1-2\alpha)\beta^\ast}\right)ES_{\beta^\ast}(L)+\frac{1-\alpha}{\alpha + (1- 2\alpha)\beta^\ast}E[L],
    \end{equation*}
    where $\beta^\ast \in \left[P\left[ L< e_{\alpha}(L) \right],P\left[ L\leq  e_{\alpha}(L) \right]\right]$.
\end{proposition}

\begin{proof}
    First, let us show that if $L$ is not identically constant, then every $\beta^\ast$ in $[P[ L< e_\alpha(L)],P[ L\leq  e_\alpha(L)]]$ is strictly between $0$ and $1$.
    Since $E[L]\leq e_\alpha(L)$ and $0<P[L<E[L]]<1$, it hold that  $0<\beta^\ast\leq 1$.
    If  $\beta^\ast=1$, then $P[ L>  e_\alpha(L)]]=0$ and hence $E[(L-e_\alpha(L))^+]=0$.
    The first order condition can be written as 
    \begin{equation}\label{eq:foc}
        e_\alpha(L)-E[L]=\frac{(2\alpha-1)E[(L-e_\alpha(L))^+]}{1-\alpha}.
    \end{equation}
    It follows that $e_\alpha(L)=E[L]$ which contradict the fact that $0<P[L\leq E[L]]<1$.
    Hence, $\beta^\ast$ must be in $(0,1)$.
    By the definition of $\beta^\ast$, it holds that $e_\alpha(L)$ is in $[q_L(\beta^\ast),q^+_L(\beta^\ast)]$.
    As a result of \citep[Proposition 4.2]{acerbi2002b}, it holds that 
    \begin{equation*}
        ES_{\beta^\ast}(L)=e_\alpha(L)+\frac{E[(L-e_\alpha(L))^+]}{1-\beta^\ast}.
    \end{equation*}
    Together with Relation \eqref{eq:foc}, this gives 
    \begin{equation*}
        ES_{\beta^\ast}(L)=e_\alpha(L)+\frac{1-\alpha}{(2\alpha-1)(1-\beta^\ast)}(e_\alpha(L)-E[L]).
    \end{equation*}
    Solving for $e_\alpha(L)$ gives the required expressions of $e_\alpha$.
\end{proof}

From the proof of Proposition \ref{prop:OGALA}, it is easy to see that  the inequality \eqref{eq:bound} become equality, that is, the optimal $\beta^\ast=\alpha$ if and only if $e_\alpha$ is the $\alpha$-quantile.
When $F_L$ is strictly increasing and continuous, it holds that $e_\alpha(L) = q_L(F_L(e_\alpha))$.
Hence, in this special case $\beta^\ast=\alpha$ if and only if the expectile is a value at risk at confidence level $\alpha$.
This is the case for instance when $q_L(\alpha)=(2\alpha-1)/\sqrt{\alpha (1-\alpha)}$, see \citet{koenker1993}.

\begin{remark}
    If $F_L$ is strictly increasing and continuous, then $\beta^\ast$ uniquely solves
    \begin{equation}\label{eq:opt}
        q_L(\beta^\ast)=\left(1-\frac{1-\alpha}{\alpha + (1-2\alpha)\beta^\ast}\right)ES_{\beta^\ast}(L)+\frac{1-\alpha}{\alpha+ (1-2\alpha)\beta^\ast}E[L].
    \end{equation} 
\end{remark}

Let $L_1,\dots,L_d$ be in $L^1$ such that $L=\sum_{k=1}^d L_k$. 
For any risk measure $R$, the Euler risk contribution of position $L_k$ to the risk capital $R(L)$ is defined as 
\begin{equation}\label{eq:euler}
    R(L_k|L):=\lim_{\varepsilon \to 0}\frac{R(L+\varepsilon L_k )-R(L)}{\varepsilon},
\end{equation}
provided that the limit exist for each $k=1,\dots d$.
It is well known, see for instance \citep{kalkbrener2005, tasche2008}, that if $R=ES_\alpha$, then
\begin{equation*}
    ES_\alpha(L_k|L)= ES_\alpha(L_k|L>q_L(\alpha))
\end{equation*}  
provided that the limit defined in \eqref{eq:euler} exists. 
Similarly, for the case where $R=e_\alpha$ such that the limit defined in \eqref{eq:euler} exist, according to \citep{emmer2015} we also get   
\begin{equation*}
    e_\alpha(L_k|L)=\frac{\alpha E[L_k 1_{\{L>e_\alpha(L)\}}]+(1-\alpha)E[L_k 1_{\{L\leq e_\alpha(L)\}}]}{\alpha +(1-2\alpha)P[L \leq e_\alpha(L)]}.
\end{equation*}
Hence, in the same sprit of Proposition \ref{prop:OGALA}, the Euler risk contributions of position $L_k$ to the expectile risk capital can also be formulated as a function of its Euler risk contribution to the expected shortfall risk capital.

\begin{proposition}\label{prop:euler}
    Let $L_1,\dots,L_d$ be in $L^1$ such that the limit defined in \eqref{eq:euler} exist for both $ES_{\beta^\ast}$ and $e_\alpha$, where $\beta^\ast=P[L\leq e_\alpha(L)]$.
    Then the Euler risk contribution of position $L_k$ to the risk capital $e_\alpha(L)$ is given by 
    \begin{equation*}
        e_\alpha(L_k|L)=\left( 1-\frac{1-\alpha}{\alpha+ (1-2\alpha)\beta^\ast} \right)ES_{\beta^\ast}(L_k|L)+\frac{1-\alpha}{\alpha +(1-2\alpha)\beta^\ast}E[L_k].
    \end{equation*}
\end{proposition}

We now turn to the question of the upper bound.
If $(\Omega,\mathcal{F},P)$ is non-atomic, from \citep{delbaen2013} and \citep{ziegel2016}, we get an other upper bound of expectile 
    \begin{equation*}
        R_{\varphi}(L):=\int_0^1 \varphi'(t)q_L(1-t)dt
    \end{equation*}
which is a distortion function corresponding to the concave distortion function $\varphi\colon [0,1]\to [0,1]$, given by $\varphi(t)=\alpha t/((2\alpha-1)t+1-\alpha)$.
Furthermore, $R_{\varphi}$ is the least one from the class of law-invariant coherent and comonotonic risk measure dominating $e_\alpha$.
It also holds that $e_\alpha(1_A)=R_{\varphi}(1_A)=\varphi(P[A])$ for each $A\in \mathcal{F}$. 
Since the upper bound given in Proposition \ref{prop:bound} is also coherent and comonotone, it follows in particular that 
\begin{equation*}
    R_{\varphi}(L)\leq \left(1-\frac{1-\alpha}{\alpha}\right)ES_\alpha(L)+\frac{1-\alpha}{\alpha}E[L].
\end{equation*}
However, the upper bound $(1-(1-\alpha)/\alpha)ES_\alpha(L)+(1-\alpha)E[L]/\alpha $ is the least one within the class of expected shortfall in the sense stated in the following proposition.
\begin{proposition}\label{prop:minimal}
    Suppose $(\Omega,\mathcal{F},P)$ be non-atomic. Then
    \begin{equation*}
        \left( 1-\frac{1-\alpha}{\alpha} \right)ES_{\alpha}(L)+\frac{1-\alpha}{\alpha} E[L]
    \end{equation*}
    is the smallest risk measure of the form $(1-\lambda)ES_\beta(L)+\lambda ES_\delta (L)$ with $0\leq \lambda\leq 1$, $0\leq \beta <1$ and $0\leq \delta < 1$ uniformly dominating $e_\alpha(L)$ for $L$ in $L^1$.
\end{proposition}
\begin{proof}
    Note that 
    \begin{equation*}
        (1-\lambda)ES_\beta(L)+\lambda ES_\delta(L)=R_{\varphi_{\lambda, \beta,\delta}}(L):= \int_0^1 \varphi_{\lambda,\beta,\delta}'(t)q_L(1-t)dt
    \end{equation*}
    for the concave distortion function 
    \begin{equation*}
        \varphi_{\lambda, \beta, \delta}(t) := (1-\lambda)\left(\frac{t}{1-\beta}\wedge 1\right) + \lambda \left( \frac{t}{1-\delta} \wedge 1\right) \quad \text{for }0\leq t\leq 1,
    \end{equation*}
    which is continuous and strictly increasing with $\varphi_{\lambda,\beta,\delta}(0)=0$ and $\varphi_{\lambda,\beta,\delta}(1)=1$.
    For $\beta=\alpha$, $\lambda=(1-\alpha)/\alpha$ and $\delta=0$, we have 
    \begin{equation*}
        R_{\varphi_{\lambda, \beta,\delta}}(L)=\left( 1-\frac{1-\alpha}{\alpha} \right)ES_{\alpha}(L)+\frac{1-\alpha}{\alpha} E[L].
    \end{equation*}

    On the one hand, let $\varphi(t^\ast)>\varphi_{\lambda, \beta,\delta}(t^\ast)$ for some $t^\ast$ in $(0,1)$. 
    Since $(\Omega,\mathcal{F})$ is non-atomic, there exist $A\in \mathcal{F}$ such that $P[A]=t^\ast$. 
    Following \citep{foellmer2016} and \citep{shapiro2013}, we get $R_{\varphi_{\lambda,\beta,\delta}}(1_A)<\varphi(P[A])=e_\alpha(1_A)$ and therefore $R_{\varphi_{\lambda,\beta,\delta}}$ can not dominate $e_\alpha$.
    Hence, for every $0 \leq \lambda \leq 1$, $0\leq  \beta < 1$ and $0\leq \delta < 1$, we have $R_{\varphi_{\lambda,\beta,\delta}}$ dominate $e_\alpha$ only if $\varphi \leq \varphi_{\lambda,\beta,\delta}$.

    On the other hand, for every $0 \leq \lambda \leq 1$, $0\leq  \beta < 1$ and $0\leq  \delta < 1$  such that $\varphi \leq \varphi_{\lambda,\beta,\delta}$, it holds 
    \begin{equation*}
        R_{\varphi_{\lambda,\beta,\delta}}(L)\geq e_\alpha(L).
    \end{equation*}
    In this case, $\varphi_{(1-\alpha)/\alpha, \alpha,0} \leq \varphi_{\lambda, \beta,\delta}$.
    Indeed, since $\varphi_{(1-\alpha)/\alpha, \alpha,0}$ is tangent to $\varphi$ at the point $(0,0)$ and $(1,1)$, and $\varphi$ is strictly concave, $\varphi_{\lambda, \beta,\delta}(t)<\varphi_{(1-\alpha)/\alpha, \alpha,0}(t)$ for some $t$ in $(0,1)$ implies there exist $t^\ast$ in $(0,1)$ such that $\varphi(t^\ast)>\varphi_{\lambda, \beta,\delta}(t^\ast)$ for some $t^\ast$ in $(0,1)$, see Figure \ref{fig:primitive}.
    By a similar argument, it follows that $ R_{\varphi_{\lambda,\beta,\delta}}$ can not dominate $e_\alpha$. 
    Therefore, $R_{\varphi_{\lambda,\beta,\delta}}$ is the better one when $\lambda=(1-\alpha)/\alpha$, $\beta=\alpha$ and $\delta=0$.
\end{proof}
\begin{figure}[H]
    \centering
    \includegraphics[width=0.8\textwidth]{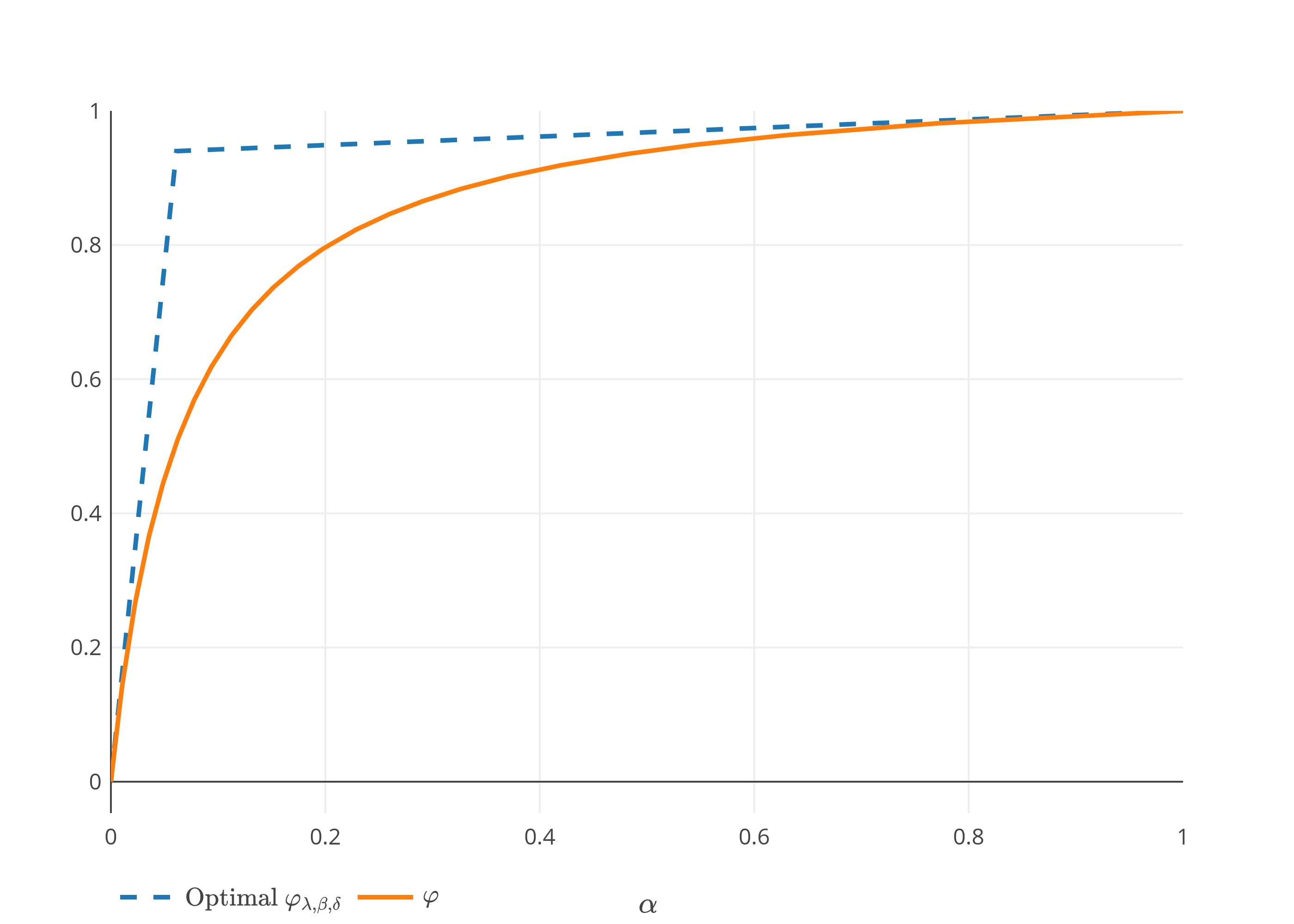}
    \caption{Graph of $\varphi$ and optimal $\varphi_{\lambda,\beta,\delta}$ for $\alpha=94\%$.}
    \label{fig:primitive}
\end{figure}

\begin{remark}
    If  $(\Omega,\mathcal{F},P)$ contains atoms, Proposition \ref{prop:minimal} may not be true in general.
    For instance, when $\Omega=\{\omega_1,\omega_2\}$ with $P[\omega_1]=P[\omega_2]=0.5$, $\alpha=9/10$, $\lambda=\delta=0$ and  $\beta=4/9$, we get $\varphi(t)=9t/(8t+1)$,
    \begin{equation*}
        \varphi_{\lambda, \beta,\delta}(t)=\begin{cases}
            \frac{9t}{5} & \text{for }0\leq t\leq \frac{5}{9}\\
            1 & \text{for }\frac{5}{9}<t\leq 1           
        \end{cases} \quad \text{and} \quad \varphi_{(1-\alpha)/\alpha, \alpha,0}(t)=\begin{cases}
            9t & \text{for }0\leq t\leq \frac{1}{9}\\
            \frac{t+8}{9} & \text{for }\frac{1}{9}<t\leq 1
        \end{cases}.
    \end{equation*}
    Every $L$ in $L^1$ is of the form $L=x_1 1_{\{\omega_1\}}+x_2 1_{\{\omega_2\}}$ for some $x_1$ and $x_2$ in $\mathbb{R}$.
    Without loss of generality, assume $x_1\leq x_2$, a simple computation yields 
    $e_\alpha(L)=ES_\beta(L) =R_{\varphi}(L)=0.9(x_2-x_1)+x_1$ and $(1-(1-\alpha)/\alpha)ES_{\alpha}+(1-\alpha) E[L]/\alpha = 17(x_2-x_1)/18+x_1$.
    This implies that  $ES_\beta$ dominates $e_\alpha$, but it is dominated by $(1-(1-\alpha)/\alpha)ES_{\alpha}+(1-\alpha) E[L]/\alpha$.
    Hence, $(1-(1-\alpha)/\alpha)ES_{\alpha}+(1-\alpha) E[L]/\alpha$ is not the least one.
\end{remark}

\section{Expectile Versus Expected Shortfall: Asymptotic Comparison}\label{sec:04}
Throughout this section, we consider a loss profile $L$ with zero mean and for ease of notation its cumulative distribution is denoted by $F$.
For ease of notations, we also use 
\begin{equation*}
    e_{\alpha} := e_\alpha(L), \quad q_{\alpha} := q_L(\alpha) \quad  \text{and}\quad ES_{\alpha} := ES_{\alpha}(L).
\end{equation*}
Given $L_1, \ldots, L_n$, independent copies of $L$, we denote by $F_n$ the empirical cumulative distribution function of the empirical measure $(\sum_{k=1}^n \delta_{L_k})/n$.
Correspondingly, we use the notations $q_{\alpha,n}$, $ES_{\alpha, n}$ and $e_{\alpha,n}$ for the value at risk, expected shortfall and expectile of the empirical measure, respectively.
We also use the standard notation $f(\alpha)\sim g(\alpha)$ and $f(\alpha) = o(g(\alpha))$ as $\alpha$ goes to $1$ which means that $\lim_{\alpha\nearrow 1}f(\alpha)/g(\alpha)=1$ and $\lim_{\alpha \nearrow 1} f(\alpha)/g(\alpha) = 0$, respectively.
For $f(\alpha)\sim g(\alpha)$ we may equivalently write $f(\alpha)=g(\alpha)+o(g(\alpha))$. 

For a given risk level $\alpha$, expectile and value at risk are less conservative than expected shortfall, that is, $e_\alpha\leq ES_\alpha$ and $q_\alpha \leq ES_\alpha$.
Expectile can be less or more conservative than value at risk depending on the considered loss profile, see \citep{bellini2017}.
When $F$ is in the maximum domain of attraction of an extreme value distribution function, \citep{bellini2017} and \citep{mao2015} give asymptotic comparison between value at risk and expectile. 
For Fr\'{e}chet type $MDA(\Phi_\eta)$ with $\eta>1$, from \citep{bellini2017} we have the relation 
\begin{equation}\label{eq:expectileQ}
    e_\alpha \sim (\eta-1)^{\frac{-1}{\eta}}q_\alpha, \quad \alpha \nearrow 1.
\end{equation}
Using this asymptotic result, \citep{daouia2018} introduces the extreme expectile estimator 
\begin{equation*}
    \hat{e}_\alpha :=(\hat{\eta}-1)^{\frac{-1}{\hat{\eta}}}q_{\alpha,n}, 
\end{equation*}
where $\hat{\eta}$ is the Hill estimator of $\eta$.
Relation \eqref{eq:expectileQ} also allows for instance to provide estimation of tail index using the ratio of empirical expectile to empirical value at risk for large values of $\alpha$.
However, when looking at the tails, extreme quantile estimation may needs  extremely very large sample size compared to expected shortfall and expectile in the case of heavy tail distributions as $\alpha$ close to $1$.
This fact is explained in Proposition \ref{prop:concentration} and \ref{prop:samplesize} below and this motivates the study of asymptotic comparison of expectile with respect to expected shortfall instead of value at risk.
Therefore, in this section we are comparing the number of sample size required for the estimation of extreme quantile with respect to expected shortfall and expectile when the sample is taken from distribution $F$ that belongs to a Fr\'{e}chet type $MDA(\Phi_\eta)$.
We also consider the asymptotic behavior of expectile with respect to expected shortfall and the optimal $\beta^\ast$ when $F$ belongs to extreme value distributions as the confidence level $\alpha$ goes to $1$.
These propositions are based on concentration inequalities using Wasserstein distance.
Such an approach has recently been adopted by \citet{daniel2020} to provide error bounds for classes of empirical risk measures.
In our context we focus however on the sensitivity with respect to the confidence level $\alpha$ and make an approach which is slightly different using the results in \citet{fournier2015}.

The asymptotic comparison uses techniques from extreme value theory.
We say $F$ is in the maximum domain of attraction of an extreme value distribution function $H$, denoted by $MDA(H)$, if 
\begin{equation*}
    \lim_{n\nearrow \infty} F^n(c_n x+d_n)=H(x)
\end{equation*}
for some constants $c_n>0$ and $d_n\in \mathbb{R}, n\in \{1,2,\dots\}$.
It is well known that extreme value distribution $H$ belongs to either one of the following three categories:
Weibull\footnote{$\Psi_\eta(x)=\exp(-(-x)^\eta)$ for $x<0$.} ($\Psi_\eta$), Gumbel \footnote{$\Lambda(x)=\exp(-e^{-x})$ for $x \in \mathbb{R}$.} ($\Lambda$) or Fr\'{e}chet\footnote{$\Phi_\eta(x)=\exp(-x^{-\eta})$ for $x>0$.} ($\Phi_\eta$), where $\eta>0$, see \citep{neil2015,tang2012,mao2012,bellini2017} for more discussion in the present context.

Let $U(t):=q_{1-1/t}$ for $t>1$.
The condition that $F$ belongs to the maximum domain of attraction can be equivalently given by the extended regular variation of $U$.
Recall that a measurable function $f:\mathbb{R}_+\to \mathbb{R}$ is said to be of extended regular variation with parameter $\eta\in \mathbb{R}$, denoted by $f\in ERV_\eta$, if there exist a function $a:\mathbb{R}_+\to \mathbb{R}_+$ such that for each $x>0$, 
\begin{equation*}
    \lim_{t\nearrow \infty}\frac{f(tx)-f(t)}{a(t)}=\begin{cases}
        \frac{x^\eta-1}{\eta}, &  \eta\neq 0\\
        \ln{x}, & \eta=0
    \end{cases}.
\end{equation*}
It is known that $F$ is in the maximum domain of attractions of Fr\'{e}chet type $MDA(\Phi_\eta)$, with $\eta>0$  if and only if  $U\in ERV_{\frac{1}{\eta}}$. 
$F$ is in the maximum domain of attractions of Weibull type $MDA(\Psi_\eta)$, with $\eta>0$ if and only if $U\in ERV_{-\frac{1}{\eta}}$.
Finally, $F$ is in the maximum domain of attraction of Gumbel type $MDA(\Lambda)$ if and only if  $U\in ERV_{0}$, see \citep[Theorem 1.1.6]{dehaan2006} for instance.

The Wasserstein distance between the probability measures with cumulative distributions $F_n$ and $F$ is defined as 
\begin{equation*}
    w(F_n,F)= \inf\{E[|Z-Y|]\colon Z\sim F_n \quad \text{and} \quad Y\sim F\}.
\end{equation*}
We consider the following assumptions on $L$:
\begin{align}
    \mathcal{E}_{k,r}(L) &  < \infty \quad \text{for some}\quad  k>1, r>0  \quad \text{or} \label{eq:assumption01}\\
    \mathcal{E}_{k,r}(L) & < \infty \quad \text{for some}\quad  k\in (0,1),  r>0, \label{eq:assumption02}\\
    m_q(L) & <\infty  \quad \text{for some}\quad  q>2 \label{eq:assumption03}
\end{align}
where $m_q(L)= E[|L|^q]$ and $\mathcal{E}_{k,r}(L)=E[\exp(r|L|^k)]$.
A simple application of concentration results by \citet{fournier2015} yields the following concentration inequalities for expected shortfall and expectile.
\begin{proposition}\label{prop:concentration}
    Let $L_1,\dots,L_n$ be a random sample from $F$ such that either assumption \eqref{eq:assumption01}, \eqref{eq:assumption02} or \eqref{eq:assumption03} holds.
    If $1/2\leq \alpha<1$, for all $n\geq 1$ and $0<\varepsilon \leq \frac{\alpha}{1-\alpha}$, it holds 
    \begin{equation*}
        P\left[ \left| ES_{\alpha,n} - ES_{\alpha} \right|\geq \varepsilon \right]\leq B(n,\varepsilon(1-\alpha))
    \end{equation*}
    and 
    \begin{equation*}
        P\left[ \left| e_{\alpha,n} - e_{\alpha} \right|\geq \varepsilon \right]\leq B\left(n,\frac{\varepsilon (1-\alpha)}{\alpha}\right)
    \end{equation*}
    where
    \begin{equation*}
        B(n,h)=C
        \begin{cases}
            \displaystyle \exp\left(-cn h^2\right) & \text{under assumption \eqref{eq:assumption01}}\\
            \displaystyle \exp\left(-cn^{s}h^2\right) & \text{for all }0<s<k \text{ under assumption \eqref{eq:assumption02}}\\
            \displaystyle n^{1-s}h^{2(1-s)} & \text{for all }2<s<q \quad \text{ under assumption \eqref{eq:assumption03}}
        \end{cases}
    \end{equation*}
    where $C$ and $c$ are positive constant that depends only on $k$, $r$, $q$, $s$, $\mathcal{E}_{k,r}$ and $m_q$ depending on the set of assumptions.
    
    It follows that the sample sizes required for expected shortfall and expectile with a given precision $\varepsilon$ and confidence $\gamma$ is given by
    \begin{equation*}
        n_{ES_\alpha} = H(\gamma, \varepsilon(1-\alpha)) \quad \text{and} \quad n_{e_\alpha} = H\left(\gamma, \frac{\varepsilon (1-\alpha)}{\alpha}\right),
    \end{equation*}
    respectively, where
    \begin{equation*}
       H(\gamma,h)=
        \begin{cases}
            \frac{-\ln\left(\frac{\gamma}{C}\right)}{c} h^{-2} & \text{under assumption \eqref{eq:assumption01}}\\
            \left(\frac{-\ln\left(\frac{\gamma}{2C}\right)}{c}\right)^{\frac{1}{s}}h^{-\frac{2}{s}} & \text{for all } 0< s< k \text{ under assumption \eqref{eq:assumption02}}\\
            \left( \frac{C}{\gamma} \right)^{\frac{1}{s-1}} h^{-2} & \text{for all } 2<s<q\text{ under assumption \eqref{eq:assumption03}}
        \end{cases}
    \end{equation*}
\end{proposition}
\begin{proof}
    According to \citep[Theorem $2$]{fournier2015}, for all $n\geq 1$ and $0<\varepsilon \leq 1$ it holds\footnote{For the third case, it is possible to have $0<s<q$. However the case where $0<s\leq 2$ is not the best choice in terms of bounds.} 
    \begin{multline}\label{eq:wasserstein}
        P[w(F_n, F)\geq\varepsilon]\\
        \leq C 
        \begin{cases}
            \exp(-cn\varepsilon^2) & \text{under assumption \eqref{eq:assumption01}}\\
            \exp(-cn\varepsilon^2)+\exp(-cn^s\varepsilon^{s}) & \text{for all }0<s<k\text{ under assumption \eqref{eq:assumption02}}\\
            \exp(-cn\varepsilon^2)+ n^{1-s}\varepsilon^{-s} & \text{for all }2<s<q \text{ under assumption \eqref{eq:assumption03}}
        \end{cases}
    \end{multline}
    Note that the constants $C$ and $c$ are positive and depends only on $r$, $k$, $s$ and $\mathcal{E}_{k,r}(L)$ for the case  \eqref{eq:assumption01},  on $r$, $k$ and $\mathcal{E}_{k,r}(L)$ for the case  \eqref{eq:assumption02} and on  $q$ and $m_q(L)$ for the case \eqref{eq:assumption03}, see \citep{fournier2015}.
    Let us treat the cases separately.
    \begin{itemize}[fullwidth]
        \item Under assumption \eqref{eq:assumption01}, if holds immediately that 
            \begin{equation*}
                P[w(F_n, F)\geq\varepsilon]\leq C \exp\left(-cn\varepsilon^2 \right) .
            \end{equation*}
        \item Under assumption \eqref{eq:assumption02}, since $0<s<k<1$, $n\geq 1$ and $\varepsilon\leq 1$, it follows that $n\varepsilon^2 \geq n^{s}\varepsilon^2$ and $n^s \varepsilon^s\geq n^s \varepsilon^2$.
            Hence
            \begin{equation*}
                P[w(F_n, F)\geq\varepsilon]\leq 2C \exp\left(-cn^{s}\varepsilon^2 \right) \quad \text{for all} \quad 0<s<k.
            \end{equation*}
        \item Under assumption \eqref{eq:assumption03}, for all $2<s<q$, since $\varepsilon \leq 1$ and $s-2>0$, it holds 
            \begin{align*}
                P[w(F_n, F)\geq\varepsilon] & \leq  C n^{1-s}\varepsilon^{2(1-s)}\left(\varepsilon^{s-2}+n^{s-1}\varepsilon^{2(s-1)}\exp\left(-cn \varepsilon^2\right)\right)\\
                                            & \leq Cn^{1-s}\varepsilon^{2(1-s)}\left( 1 + n^{s-1}\varepsilon^{2(s-1)}\exp\left(-cn \varepsilon^2\right)\right) 
            \end{align*}
            The function $x \mapsto x^{s-1}\varepsilon^{2(s-1)}\exp\left(-cx \varepsilon^2\right)$ is maximum at $x=(s-1)/(c\varepsilon^2)>0$.
            Plugging back this maximum, we get 
            \begin{equation*}
                P[w(F_n, F)\geq\varepsilon]\leq  C n^{1-s}\varepsilon^{2(1-s)}\left( 1+\frac{(s-1)^{s-1}}{c^{s-1}}e^{1-s} \right)
            \end{equation*}
            the right hand side being independent of $\varepsilon$ and $n$.
    \end{itemize}
    All together it follows that
    \begin{equation}\label{eq:wasserstein01}
        P[w(F_n, F)\geq\varepsilon]\leq B(n,\varepsilon).
    \end{equation}
    Since $w(F_n,F)=\int_0^1|q_{u,n}-q_u|du$, see \citep{major1978} for instance, for each $\alpha$ in $(0,1)$ the definition of expected shortfall gives 
    \begin{equation}\label{eq:wasserstein02}
        |ES_{\alpha,n}-ES_{\alpha}|=\frac{1}{1-\alpha}\left|\int_{\alpha}^1 (q_{u,n}-q_u)du\right| \leq \frac{1}{1-\alpha}w(F_n,F).
    \end{equation}
    For each $\alpha$ in $[1/2, 1)$ from \citep{bellini2014}, it also holds
    \begin{equation}\label{eq:wasserstein03}
        |e_{\alpha,n}-e_\alpha|\leq \frac{\alpha}{1-\alpha}w(F_n,F).
    \end{equation}
    The Relations \eqref{eq:wasserstein01} -\eqref{eq:wasserstein03} yields the required concentration inequalities.
    As for the sample size fix a confidence level $\gamma$ such that $B(n, \varepsilon (1-\alpha)) \leq \gamma$ for expected shortfall and $B\left(n, \frac{\varepsilon (1-\alpha)}{\alpha}\right) \leq \gamma$ for the expectile and solving for $n$ yields the required result.
\end{proof}

According to Proposition \ref{prop:concentration}, for a given precision $\varepsilon$ and confidence $\gamma$, it follows that $n_{ES_\alpha}$ and $n_{e_\alpha}$ tends to $\infty$.
Furthermore, under either of the assumptions \eqref{eq:assumption01}, \eqref{eq:assumption02} or \eqref{eq:assumption03}, it holds 
\begin{equation*}
    n_{ES_\alpha} \sim n_{e_\alpha} \quad \alpha \nearrow 1.
\end{equation*}
However, this is not  the case in general for value at risk versus expected shortfall and  value at risk versus expectile.
For instance, suppose $F$ has a strictly positive and continuous density function $f$ on the interval $[q_\alpha-\delta, q_\alpha+\delta]$ for some $\delta>0$.
According to  \citep[Corollary $2.1$]{gao2011} for every $n\geq 1$ and $\varepsilon$ in $(0,\delta]$, it holds  
\begin{equation*}
    P[|q_{\alpha,n}-q_\alpha|\geq \varepsilon]\leq 4 \exp(-2n \varepsilon^2 \delta^2_\alpha),
\end{equation*}
where 
\begin{equation*}
    \delta_\alpha = \inf_{x\in [q_\alpha-\delta, q_\alpha+\delta]}f(x).
\end{equation*}
It implies that sample size for a given precision $\varepsilon$ in $(0,\delta]$ and confidence $\gamma$ is given by 
\begin{equation}\label{eq:quantileconcent}
    n_{q_\alpha} = \frac{-\ln\left(\frac{\gamma}{4}\right)}{2 \varepsilon^2 }\delta_{\alpha}^{-2}.
\end{equation}
Hence, $n_{q_\alpha}$ depends on the relative behavior of $\delta_\alpha$ as $\alpha$ goes to $1$.
Assuming that the density $f$ become non-increasing for large enough values, we get $\delta_\alpha=f(q_\alpha+\delta)$.
When $F$ belongs to the Fr\'{e}chet type $MDA(\Phi_\eta)$, it holds that $1-F$ is in $RV_{-\eta}$.
From \citep[Proposition B.1.9]{dehaan2006}, we get $\lim_{x\nearrow \infty} xf(x)/(1-F(x))=\eta$ and this implies that  
\begin{equation}\label{eq:asymptotic}
    \frac{\delta_\alpha}{1-\alpha}=\frac{f(q_\alpha+\delta)}{1-F(q_\alpha)}\sim \frac{q_\alpha f(q_\alpha)}{1-F(q_\alpha)}\frac{1}{q_\alpha} \to 0, \quad \alpha \nearrow 1.
\end{equation}
Hence, asymptotically for heavy tailed distributions, we may have $\delta_\alpha =o(1-\alpha)$ and $n_{q_\alpha}$ tends to  $\infty$ faster than  $n_{ES_\alpha}$ and $n_{e_\alpha}$  as shown in the following proposition.
\begin{proposition}\label{prop:samplesize}
    Suppose that $F$ belongs to a Fr\'{e}chet type $MDA(\Phi_\eta)$ with $\eta>1$ and some moment $q>2$.
    Suppose further that the density function $f$ of $F$ is strictly positive and decreasing for large enough values.
    Then, for every $\varepsilon \leq \alpha/(1-\alpha)$, and confidence level $\gamma$, as $\alpha$ goes to $1$ it holds
    \begin{equation*}
        \frac{n_{ES_\alpha}}{n_{q_\alpha}}=o(1), \quad \text{and} \quad  \frac{n_{e_\alpha}}{n_{q_\alpha}}=o(1).
    \end{equation*}
\end{proposition}
\begin{proof}
    Since $q_\alpha$ goes to $\infty$ as $\alpha$ goes to $1$, by assumption, $f$ is strictly positive and continuous on the interval $[q_\alpha-1,q_\alpha+1]$.
    A simple application of Proposition \ref{prop:concentration} and Relation \eqref{eq:quantileconcent} yields  
    \begin{equation*}
        \frac{n_{ES_\alpha}}{n_{q_\alpha}}= \tilde{C}\frac{\delta_{\alpha}^2}{(1-\alpha)^2} 
    \end{equation*}
    for some constant $\tilde{C}$.
    As a result of Relation \eqref{eq:asymptotic}, $n_{ES_\alpha}/n_{q_\alpha}$ goes to $0$ as  $\alpha$ goes to $1$.
    Likewise, $n_{e_\alpha}/n_{q_\alpha}=o(1)$ as $\alpha$ goes to $1$. 
\end{proof}

We now turn to the asymptotic behavior of the optimal $\beta^\ast$ and expectile with respect to expected shortfall.
A direct application of \citep{bellini2014, bellini2017, mao2012} yields
\begin{proposition}\label{prop:frechet}
    For $F$ in the maximum domain of attraction of Fr\'{e}chet type $MDA(\Phi_\eta)$ with $\eta>1$, as the confidence level $\alpha$ goes to $1$, we have that
    \begin{equation*}
        \frac{1-\beta^\ast}{1-\alpha}\sim \eta-1, \quad \text{and}\quad e_\alpha\sim \frac{(\eta-1)^{\frac{\eta-1}{\eta}}}{\eta}ES_\alpha.
    \end{equation*} 
\end{proposition}
\begin{proof}
    The relation between $\beta^\ast$ and $\alpha$ is given in \citep{bellini2014}, and  from \citep{mao2012} we have
    \begin{equation*}
        ES_\alpha\sim \frac{\eta}{\eta-1}q_\alpha.
    \end{equation*}
    Together with Relation \eqref{eq:expectileQ} this yields the required result.
\end{proof}
Beyond this first order expansion, a second order one is useful to determine the rate of convergence. 
In order to do so, we impose a second-order regular variation condition on $F$.
A measurable function $f:\mathbb{R} \to \mathbb{R}$ is said to be of regular variation with parameter $\eta\in \mathbb{R}$, denoted by $f\in RV_\eta$, if $\lim_{t\nearrow \infty}\frac{f(tx)}{f(t)}=x^\eta$, for each $x\in \mathbb{R}$. 
A regularly varying function $f:\mathbb{R}\to \mathbb{R}$ which is eventually positive is said to be of second-order regular variation with first-order parameter $\eta\in \mathbb{R}$ and second-order parameter $\rho\leq 0$, denoted by $f\in 2RV_{\eta,\rho}$, if $f\in RV_\eta$ and there exists a measurable function $A(t)$ which does not change sign eventually and converges to $0$ as $t$ goes to $\infty$ such that, for each $x>0$ 
\begin{equation*}
    \lim_{t\nearrow\infty}\frac{f(tx)/f(t)-x^\eta}{A(t)}=
    \begin{cases}
        x^\eta \frac{x^\rho-1}{\rho}& \text{ if }\rho\neq 0\\
        x^\eta \ln{x}& \text{ if }\rho=0
    \end{cases}
\end{equation*}
see, \citep{dehaan2006,maolv2012} for further properties of regular variations.
The function $A$ is in $RV_{\rho}$, see \citep[Theorem 2.3.3]{dehaan2006}, and is called the auxiliary function for $f$. 

\begin{proposition}\label{prop:frechet2}
    For $F$ in the maximum domain of attraction of Fr\'{e}chet type $MDA(\Phi_\eta)$ such that $1-F\in 2RV_{-\eta,\rho}$ with $\eta>1$, $\rho\leq 0$ and auxiliary function $A$, as the confidence level $\alpha$ goes to $1$, it holds
    \begin{equation*}
     \frac{e_\alpha}{ES_\alpha}=
     \frac{(\eta-1)^{\frac{\eta-1}{\eta}}(2\alpha-1)^{\frac{1}{\eta}}}{\eta}\left(1-C_{\eta,\rho}A(q_\alpha)+o\left(A(q_{\alpha})\right) \right),
 \end{equation*}
 where \begin{equation*}
     C_{\eta,\rho}
     =\begin{cases}
         \frac{\eta-1}{\rho \eta}\left(\frac{1-(\eta-1)^{-\frac{\rho}{\eta}}}{\eta-\rho-1}\right), & \rho \neq 0\\
      \frac{1}{\eta}\left(\ln{(\eta-1)}+\frac{1}{\eta-1}\right), &  \rho=0
     \end{cases}.
 \end{equation*}
    Furthermore, 
    \begin{equation*}
        \frac{1-\beta^\ast}{1-\alpha}=\frac{\eta-1}{2\alpha-1}\left[1-\frac{(\eta-1)^{-\frac{\rho}{\eta}}}{\eta-\rho-1}A(q_\alpha)+o\left( A(q_\alpha) \right))\right].
    \end{equation*}
 \end{proposition}
 \begin{proof}
    Let $1-F$ be in $2RV_{-\eta,\rho}$ with $\eta>1$, $\rho\leq 0$ and auxiliary function $A$.
    According to \citep[Theorem 3.1]{mao2015a}, we get 
    \begin{equation*}
        e_\alpha=\left( \frac{2\alpha-1}{\eta-1}\right)^{\frac{1}{\eta}}q_\alpha\left[1+ B_{\eta,\rho} A(q_\alpha) + o\left( A(q_\alpha) \right)\right],
    \end{equation*}
    where 
    \begin{equation*}
        B_{\eta,\rho}=
        \begin{cases}
            \frac{1}{\eta\rho}\left[\frac{(\eta-1)^{1-\frac{\rho}{\eta}}}{\eta-\rho-1}-1\right] & \rho\neq 0\\
            -\frac{1}{\eta}\ln{(\eta-1)} & \rho =0
        \end{cases}.        
    \end{equation*}
    The condition $1-F\in 2RV_{-\eta,\rho}$ with auxiliary function $A$ is equivalent to $U$ is in $2RV_{\frac{1}{\eta},\frac{\rho}{\eta}}$ with auxiliary function $\eta^{-2}A(U)$, see \citep[Theorem 2.3.9]{dehaan2006}.
    Hence, by  \citep[Theorem 4.5]{mao2012}, we get
    \begin{equation*}
        ES_\alpha = \frac{\eta}{\eta-1}q_\alpha\left[1+\frac{1}{\eta(\eta-\rho-1)}A(q_\alpha) + o\left( A(q_\alpha) \right)\right].
    \end{equation*} 
    It follows that
    \begin{multline*}
        \frac{e_\alpha}{ES_\alpha}=
        \frac{(\eta-1)^{\frac{\eta-1}{\eta}}(2\alpha-1)^{\frac{1}{\eta}}}{\eta} \frac{1+ B_{\eta,\rho} A(q_\alpha) + o\left( A(q_\alpha) \right)}{1+\frac{1}{\eta(\eta-\rho-1)}A(q_\alpha) + o\left( A(q_\alpha)\right)}\\
        =\frac{(\eta-1)^{\frac{\eta-1}{\eta}}(2\alpha-1)^{\frac{1}{\eta}}}{\eta}\left[1+ \left(B_{\eta,\rho}-\frac{1}{\eta(\eta-\rho-1)}\right) A(q_\alpha) + o\left(A(q_\alpha)\right)\right]
    \end{multline*}
    which gives the required result for $e_\alpha/ES_\alpha$.
    
    As for the ratio of $(1-\beta^\ast)/(1-\alpha)$, from \citep[Proof of Theorem 3.1]{mao2015a}, we have 
     \begin{equation*}
         E[(L-e_\alpha)^+] = \frac{ e_\alpha (1-\beta^\ast)}{\eta-1}\left[1+\frac{1}{\eta-\rho-1}A(e_\alpha)+ o\left( A(e_\alpha) \right)\right].
     \end{equation*}
     From Relation \eqref{eq:expectileQ}, we have $e_\alpha\sim (\eta-1)^{-1/\eta}q_\alpha$.
     Since $A\in RV_\rho$, a straightforward application of \citep[Proposition B.1.10]{dehaan2006} yields $A(e_\alpha)=(\eta-1)^{-\frac{\rho}{\eta}}A(q_\alpha)+ o(A(q_\alpha))$ showing that  
     \begin{equation*}
         \frac{E[(L-e_\alpha)^+]}{e_\alpha} = \frac{1-\beta^\ast}{\eta-1}\left[1+\frac{(\eta-1)^{-\frac{\rho}{\eta}}}{\eta-\rho-1}A(q_\alpha) + o\left(A\left( q_\alpha \right)\right)\right].
     \end{equation*}
     The first order condition given by Relation \eqref{eq:foc} can also be written as 
     \begin{equation*}
         \frac{1-\alpha}{2\alpha-1} =  \frac{E[(L-e_\alpha)^+]}{e_\alpha}.
     \end{equation*}
     Combining the last two equations gives the required result for $(1-\beta^\ast)/(1-\alpha)$.
 \end{proof}

 \begin{remark}\label{rem:FR}
     When $E[L]\neq 0$, for $\tilde{L}=L-E[L]$, following \citep{mao2015a}, we get that $1-\tilde{F}$ is in $2RV_{-\eta,(-1)\vee \rho}$ with auxiliary function $A^\ast(x)=\eta E[L]/x+A(x)$, where $\tilde{F}$ is the cumulative distribution of $L-E[L]$.
 \end{remark}

 For a given confidence level $\alpha$, expectile is always less conservative than expected shortfall, that is, $e_\alpha\leq ES_\alpha$.
Proposition \ref{prop:frechet} implies that this inequality is more pronounced for Fr\'{e}chet type $MDA(\Phi_\eta)$ when $\eta$ is close to $2$.
For $\eta=2$, it holds that $ES_\alpha\sim 2e_\alpha$.
For heavy tail index $\eta$ sufficiently close to 1, $e_\alpha$ and $ES_\alpha$ are asymptotically equivalent.

According to \citep{brazauskas2008} and \citep{holzmann2016}, for a random sample from  $F$ having a finite mean  the empirical estimators  $ES_{\alpha,n}$ and $e_{\alpha,n}$ converges almost surely to $ES_\alpha$ and $e_\alpha$, respectively.
If $F$ is further a Fr\'{e}chet type $MDA(\Phi_\eta)$ with $\eta>1$, then almost surely
\begin{equation*}
    \lim_{\alpha\nearrow 1} \left(\lim_{n\nearrow \infty} \frac{e_{\alpha, n}}{ES_{\alpha,n}}\right)= \frac{(\eta-1)^{\frac{\eta-1}{\eta}}}{\eta}.
\end{equation*}
For the comparison of the empirical ratio $e_{\alpha,n}/ES_{\alpha,n}$ with the theoretical ratio $e_\alpha/ES_\alpha$, we provide a simulation study for Pareto and Student $t$ distributions with different regularity tail index $\eta$ and sample size $n$, see Example \ref{eg:Pareto} and \ref{eg:Student} for more discussion.

In the same sprit of Proposition \ref{prop:frechet}, the Euler allocations of expectile can also be formulated as a function of expected shortfall in the sense of the following proposition.
\begin{proposition}
    Let $L_1,\dots,L_d$ be a non-negative loss profiles in $L^1$ with continuous distribution such that  $L=\sum_{k=1}^d L_k$.
    If the cumulative distribution of $L_1$ belongs to a Fr\'{e}chet type $MDA(\Phi_\eta)$ with $\eta>1$ such that  
    \begin{equation*}
        \lim_{t\nearrow \infty} \frac{P[L_1>tx_1,\dots, L_d>tx_d]}{P[L_1>t]} 
    \end{equation*}
    exist for all $(x_1,\dots,x_d)\in [0,\infty]^d/\{\mathbf{0}\}$ and not identically equal to $0$, then 
    \begin{equation*}
        e_\alpha(L_k|L)\sim \frac{(\eta-1)^{\frac{\eta-1}{\eta}}}{\eta} ES_\alpha(L_k|L).
    \end{equation*}
\end{proposition}
\begin{proof}
    From Proposition \ref{prop:euler}, we have that  
    \begin{multline*}
        e_\alpha(L_k|L)=\left( 1-\frac{1-\alpha}{1-\alpha+ (2\alpha-1)(1-\beta^\ast)} \right)E[L_k|L>e_\alpha]\\+\frac{1-\alpha}{1-\alpha+ (2\alpha-1)(1-\beta^\ast)}E[L_k].
    \end{multline*}
    Under the given assumption, \citep{asimit2011} showed that $F$ is Fr\'{e}chet type $MDA(\Phi_\eta)$, $E[L_k|L>t]\sim c_k t$ as $t$ goes to $\infty$ and $ES_\alpha(L_k|L)\sim c_k q_\alpha$ for some $c_k>0$ as $\alpha$ goes to$1$.
    Together with Proposition \ref{prop:frechet}, this yields
    \begin{multline*}
        e_\alpha(L_k|L)\sim \left( 1-\frac{1}{\eta} \right) c_k e_\alpha +\frac{E[L_k]}{\eta}\sim \frac{\eta-1}{\eta} c_k e_\alpha\\\sim \frac{\eta-1}{\eta} ES_\alpha(L_k|L) \frac{e_\alpha}{q_\alpha}\sim\frac{(\eta-1)^{\frac{\eta-1}{\eta}}}{\eta} ES_\alpha(L_k|L).
    \end{multline*}
\end{proof}

We now turn to the asymptotic comparison between expectile and expected shortfall for $F$ in the domain of attraction of Weibull and Gumbel type.
For $F$ that belongs to Weibull type $MDA(\Psi_\eta)$, it is known that $\hat{x}<\infty$.
A direct application of \citep{mao2012, mao2015} yields
\begin{proposition}\label{prop:WE01}
    Let $\hat{x}:=\sup\{x:F(x)<1\}$. 
    For $F$ in the maximum domain of attraction of Weibull type $MDA(\Psi_\eta)$, as $\alpha$ goes to $1$ , it holds 
    \begin{equation*}
        1-\alpha=o(1-\beta^\ast) \quad \text{and} \quad \frac{\hat{x} - ES_\alpha}{\hat{x}-e_\alpha}=o(1).
    \end{equation*}
\end{proposition}
\begin{proof}
    The first order condition given by Relation \eqref{eq:foc} can also be re-written as 
    \begin{equation*}
        E[(L-e_\alpha)^+] =\frac{e_\alpha (1-\alpha)}{2\alpha-1}.
    \end{equation*}
    Since $e_\alpha \nearrow \hat{x}$ asymptotically the first order condition becomes 
    \begin{equation*}
        E[(L-e_\alpha)^+] \sim \hat{x} (1-\alpha).
    \end{equation*}
    From \citep[Lemma 3.2 and Remark 3.3]{mao2012} we have
    \begin{equation*}
        \frac{E[(L-x)^+]}{(\hat{x}-x)(1-F(x))}\sim \frac{1}{\eta+1}.
    \end{equation*}
    Hence, as $\alpha$ goes to $1$, it follows that
    \begin{equation*}
        \frac{E[(L-e_\alpha)^+]}{1-\beta^\ast}\sim \frac{\hat{x}-e_\alpha}{\eta+1} 
    \end{equation*}
    which implies $1-\alpha=o(1-\beta^\ast)$.
    From \citep[Theorem 3.4]{mao2012} and \citep[Proposition 3.3]{mao2015} we have
    \begin{equation*}
        \hat{x}-ES_\alpha\sim \frac{\eta}{\eta+1}(\hat{x}-q_\alpha) \quad \text{and}\quad \frac{\hat{x}-q_\alpha}{\hat{x}-e_\alpha}=o(1)
    \end{equation*}
    which yields the desired asymptotic relationship between $e_\alpha$ and $ES_\alpha$.
\end{proof}

Since the Weibull type $MDA(\Psi_\eta)$ distributions have a finite right end point $\hat{x}$, both $e_\alpha$ and $ES_\alpha$ converges to $\hat{x}$ as $\alpha$ goes to $1$.
Proposition \ref{prop:WE01} implies that $ES_\alpha$ converges to $\hat{x}$ very fast as compared to $e_\alpha$.
Under additional assumption on the distribution of $F$, we also provide a  second order expansion.
\begin{proposition}\label{prop:WE2}
    For $F$ in the maximum domain of attraction of Weibull type $MDA(\Psi_\eta)$ such that $P[L>0]>0$ and $1-F(1/\cdot)\in 2RV_{-\eta,\rho}$ with $\eta>0$, $\rho< 0$ and auxiliary function $A$, as the confidence level $\alpha$ goes to $1$, it holds 
    \begin{multline*}
        \frac{\hat{x}-ES_\alpha}{\hat{x}-e_\alpha}=\frac{\eta((2\alpha-1)(\hat{x}-q_\alpha))^{\frac{1}{\eta+1}}}{(\eta+1)C_\eta }\\
        \left[1+\left(\frac{C_\eta (\hat{x}-q_\alpha)^{\frac{\eta}{\eta+1}}}{(\eta+1) \hat{x}}+\frac{C_\eta ^{-\rho}}{\rho(\eta-\rho+1)}A_0(q_\alpha)\right)(1+o(1))\right],
    \end{multline*}
    where 
    \begin{equation*}
        C_\eta=(\hat{x}(\eta+1))^{\frac{1}{\eta+1}} \quad \text{ and }  \quad A_0(q_\alpha)=A\left((\hat{x}-q_\alpha)^{-\frac{\eta}{\eta+1}}\right).
    \end{equation*}
    Furthermore,
    \begin{equation*}
        \frac{1-\beta^\ast}{1-\alpha}\sim \left(\frac{\hat{x}(\eta+1)}{\hat{x}-q_\alpha}\right)^{\frac{\eta}{\eta+1}}.
    \end{equation*}
\end{proposition}
\begin{proof}
    Let $1-F(\hat{x}-1/\cdot)$ be in $2RV_{-\eta,\rho}$ with $\eta>0$, $\rho< 0$ and auxiliary function $A$.
    According to \citep[Proposition 2.4]{mao2012}, as $x$ goes to $\infty$ we have $1-F(\hat{x}-1/x)\sim c x^{-\eta}$ for some $c>0$.
    Hence, by \citep[Proposition 3.3]{mao2015}, it holds 
    \begin{equation}\label{eq:WE02}
        \hat{x}-e_\alpha \sim C_\eta(\hat{x}-q_\alpha)^{\frac{\eta}{\eta+1}}.
    \end{equation}
    In particular, for the same reason as in the proof of Proposition \ref{prop:frechet2}, it follows that
    \begin{equation}\label{eq:WE03}
        A\left( \frac{1}{\hat{x} - e_{\alpha}(L)} \right) \sim C_{\eta}^{-\rho}A_0\left( q_\alpha \right) \quad \text{and} \quad A\left(\frac{1}{\hat{x}-q_\alpha}\right)=o\left(A\left(\frac{1}{\hat{x}-e_\alpha}\right)\right).
    \end{equation}
    The regularity condition on $1-F(\hat{x}-1/\cdot)$ implies that $\hat{x}-U\in 2RV_{-\frac{1}{\eta},\frac{\rho}{\eta}}$ with auxiliary function asymptotically equivalent to $-\eta^{-2} A(1/(\hat{x}-U))$ as $t$ goes to $\infty$, see \citep{maolv2012}.
    Hence, using Relation \eqref{eq:WE03}, and \citep[Theorem 4.5]{tang2012}  gives 
    \begin{equation}\label{eq:WE04}
        \begin{split}
            \hat{x}-ES_\alpha & = \frac{\eta(\hat{x}-q_\alpha)}{\eta+1}\left[1-\frac{A\left(\frac{1}{\hat{x}-q_\alpha}\right)}{\eta(\eta-\rho+1)}(1+o(1))\right]\\
                                 & = \frac{\eta(\hat{x}-q_\alpha)}{\eta+1}\left[1+o\left(A\left(\frac{1}{\hat{x}-e_\alpha}\right)\right)\right].
        \end{split}
    \end{equation}
    Using Relation \eqref{eq:WE03} and \citep[Relation 3.14 and  3.17]{mao2015}, we also get that
    \begin{equation}\label{eq:WE4}
        e_\alpha=\frac{(2\alpha-1)(\hat{x}-e_\alpha)^{\eta+1}}{(\eta+1)(\hat{x}-q_\alpha)^{\eta}}\left[1+\frac{1}{\rho}\left(\frac{\eta+1}{\eta-\rho+1}\right)A\left(\frac{1}{\hat{x}-e_\alpha}\right) + o\left( A\left(\frac{1}{\hat{x}-e_\alpha}\right) \right) \right].
    \end{equation}
    Substituting the left hand side of Equation \eqref{eq:WE4} with \eqref{eq:WE02} and solving for $\hat{x}-e_\alpha$ gives
    \begin{multline}\label{eq:WE6}
        \frac{(2\alpha-1)^{\frac{1}{\eta+1}}}{C_\eta (\hat{x}-q_\alpha)^{\frac{\eta}{\eta+1}}}\left(\hat{x}-e_\alpha\right)\\
        =\left[1- \left(\frac{C_\eta(\hat{x}-q_\alpha)^{\frac{\eta}{\eta+1}}}{\hat{x}}+
            \frac{\eta+1}{\rho\left(\eta-\rho+1\right)}A\left(\frac{1}{\hat{x}-e_\alpha}\right)\right)(1+o(1))\right]^{\frac{1}{\eta+1}}\\
            =1- \left(\frac{C_\eta(\hat{x}-q_\alpha)^{\frac{\eta}{\eta+1}}}{(\eta+1)\hat{x}}+
            \frac{1}{\rho\left(\eta-\rho+1\right)}A\left(\frac{1}{\hat{x}-e_\alpha}\right)\right)(1+o(1)).
    \end{multline}
    Computations using Relations \eqref{eq:WE03}-\eqref{eq:WE04} and \eqref{eq:WE6} yields the required expression of $(\hat{x}-ES_\alpha)/(\hat{x}-e_\alpha)$.

    As for $(1-\beta^\ast)/(1-\alpha)$, using the fact that $e_\alpha$ goes to $\hat{x}$ as $\alpha$ goes to $1$, \citep[Relation 3.13]{mao2015a} and  the first order condition given by Relation \eqref{eq:foc} implies that
    \begin{equation*}
        \frac{E[(L-e_\alpha)^+]}{\hat{x}-e_\alpha}\sim \frac{1-\beta^\ast}{\eta+1} \quad \text{and}\quad \frac{\hat{x}}{\hat{x}-e_\alpha}\sim \frac{E[(L-e_\alpha)^+]}{\hat{x}-e_\alpha}.
    \end{equation*}
    Combining these Relations together with Relation \eqref{eq:WE02} yields the result.
\end{proof}
\begin{remark}\label{rem:WE}
    When $E[L]\neq 0$, the proof of Proposition \ref{prop:WE2} allow to derive the expression 
    \begin{multline*}
        \frac{\hat{x}-ES_\alpha}{\hat{x}-e_\alpha}=\frac{\eta ((2\alpha-1)(\hat{x}-q_\alpha))^{\frac{1}{\eta+1}} }{(\eta+1)\tilde{C}_\eta }\\
        \left[1+\left(\frac{\tilde{C}_\eta (\hat{x}-q_\alpha)^{\frac{\eta}{\eta+1}}}{(\eta+1) \left( \hat{x}-E[L] \right)}+\frac{\tilde{C}_\eta ^{-\rho}}{\rho(\eta-\rho+1)}A_0(q_\alpha)\right)(1+o(1))\right]
    \end{multline*}
    where
    \begin{equation*}
        \tilde{C}_\eta=((\hat{x}-E[L])(\eta+1))^{\frac{1}{\eta+1}}.
    \end{equation*}
\end{remark}
A direct combination of results by \citep{tang2012,mao2012} and \citep{bellini2017} yields 
\begin{proposition}
    For $F$ in the domain of attraction of Gumbel type $MDA(\Lambda)$, as the confidence level $\alpha$ goes to $1$, it holds $1-\alpha=o(1-\beta^\ast)$. 
    If further $F(x)=1-\exp{(-x^\tau g(x))}$ with $g\in RV_0$ and $\tau>0$, then
    \begin{equation*}
        \ln{(e_\alpha)}\sim \ln{(ES_\alpha)}.
    \end{equation*}
    Moreover, if  
    \begin{equation}\label{eq:GU01}
        \lim_{x \nearrow \infty}\left(\frac{g(c x)}{g(x)}-1\right) \ln{g(x)}=0
    \end{equation}
    for some constant $c>0$, then 
    \begin{equation*}
        e_\alpha\sim ES_\alpha.
    \end{equation*}
\end{proposition}
\begin{proof}
    From \citep[Proposition 3.6]{mao2015}, we have $1-F(q_\alpha)=o(1-\beta^\ast)$.
    For $F$ in $MDA(\Lambda)$, it is known that $1-F(q_\alpha)\sim 1-\alpha$, see \citep{tang2012} for instance.
    Therefore, $1-\alpha=o(1-\beta^\ast)$.
    As for the relationship between expectile and expected shortfall it is a direct consequence of \citep{tang2012,mao2012} and \citep[Proposition 2.4]{bellini2017}.
\end{proof}

As an application of the asymptotic results, we compare $e_\alpha$ and the upper bound
\begin{equation*}
    \left(1-\frac{1-\alpha}{\alpha}\right)ES_\alpha(L)+\frac{1-\alpha}{\alpha}E[L]
\end{equation*} 
when $F$ belongs to the domain of attractions of extreme value distributions.
In general, this bound is asymptotically equivalent to $ES_\alpha$.
For $F$ in the domain of attraction of Fr\'{e}chet type $MDA(\Phi_\eta)$ with $\eta>1$, it holds that $e_\alpha/ES_\alpha<1$ as $\alpha$ goes to $1$.
In this particular case, the bound is not asymptotically equivalent to $e_\alpha$.
Every distribution $F$ in the Weibull type  $MDA(\Psi_\eta)$ has finite end point $\hat{x}$.
This implies that both $e_\alpha$ and $ES_\alpha$ converges to $\hat{x}$.
Hence, the bound become asymptotically equivalent to $e_\alpha$ provided that $\hat{x}\neq 0$.
For $F$ in the Gumbel type $MDA(\Lambda)$ with finite end point $\hat{x}$ or satisfying condition \eqref{eq:GU01}, the bound also becomes asymptotically equivalent to $e_\alpha$. 

\section{Examples and Simulations}
For many common distributions explicit or semi-explicit expressions for both the quantile and expected shortfall are known.
Taking this advantage, in this section we illustrate the explicit or semi-explicit computations of expectile using the optimal $\beta^\ast$ and illustrating some of the results of Section \ref{sec:04} for Beta, exponential, Pareto and Student $t$ distributions.
While Beta distribution is a Weibull type $MDA(\Psi_1)$, the exponential is Gumbel type $MDA(\Lambda)$.
The Pareto and Student $t$ distributions are Fr\'{e}chet type $MDA(\Phi_\eta)$ with $\eta=a$ and $\eta=\nu$, respectively.
We also include  a simulation study for Pareto and Student $t$ distributions to compare the sample size required for empirical ratio $e_{\alpha,n}/q_{\alpha,n}$ and $e_{\alpha,n}/ES_{\alpha,n}$ to converge the theoretical ratio  $e_\alpha/q_\alpha$ and $e_\alpha/ES_\alpha$, respectively.

\begin{example}[Beta]
    For $a>0$, let $ F_L(x)=x^a$ with $x$ in $[0,1]$. Then 
    \begin{equation*}
        q_L(\beta^\ast)={\beta^\ast}^{1/a}, \quad E[L]=\frac{a}{a+1}\quad \text{and}\quad ES_{\beta^\ast}(L) =\frac{a\left(1-{\beta^\ast}^{\frac{1}{a}+1}\right)}{(1-\beta^\ast)(a+1)}.
    \end{equation*}
    Relation \eqref{eq:opt} gives the optimal $\beta^\ast$ solving 
    \begin{equation*}
        {\beta^\ast}^{\frac{1}{a}}(\alpha(a+1)+(1-2\alpha)\beta^\ast)=a\alpha.
    \end{equation*}
    Hence,
    \begin{equation*}
        e_\alpha(L)={\beta^\ast}^{1/a}.
    \end{equation*} 
    For $a=1$, the Beta distribution coincides with the standard uniform distribution and it holds that 
    \begin{equation*}
        \beta^\ast=\frac{\sqrt{\alpha(1-\alpha)}-\alpha}{1-2\alpha} = e_\alpha(L).
    \end{equation*}
    If $a\neq 1$, then $1-F_L(1/\cdot)\in 2RV_{-1,-1}$ with auxiliary function $A(x)=(a-1)x^{-1}/2$, see \citep{mao2012,mao2015} for instance.
    By Remark \ref{rem:WE}, we have
    \begin{multline*}
    \frac{1-ES_\alpha(L)}{1-e_\alpha(L)}=\frac{\sqrt{(a+1)(2\alpha-1)(1-\alpha^{\frac{1}{a}})}}{2\sqrt{2}}\\
        \left(1+\frac{a+2}{3}\sqrt{\frac{2(1-\alpha^{\frac{1}{a}})}{a+1}}(1+o(1))\right).
    \end{multline*}
\end{example}
\begin{figure}[H]
    \centering
    \begin{subfigure}{.49\textwidth}
        \includegraphics[width=\textwidth]{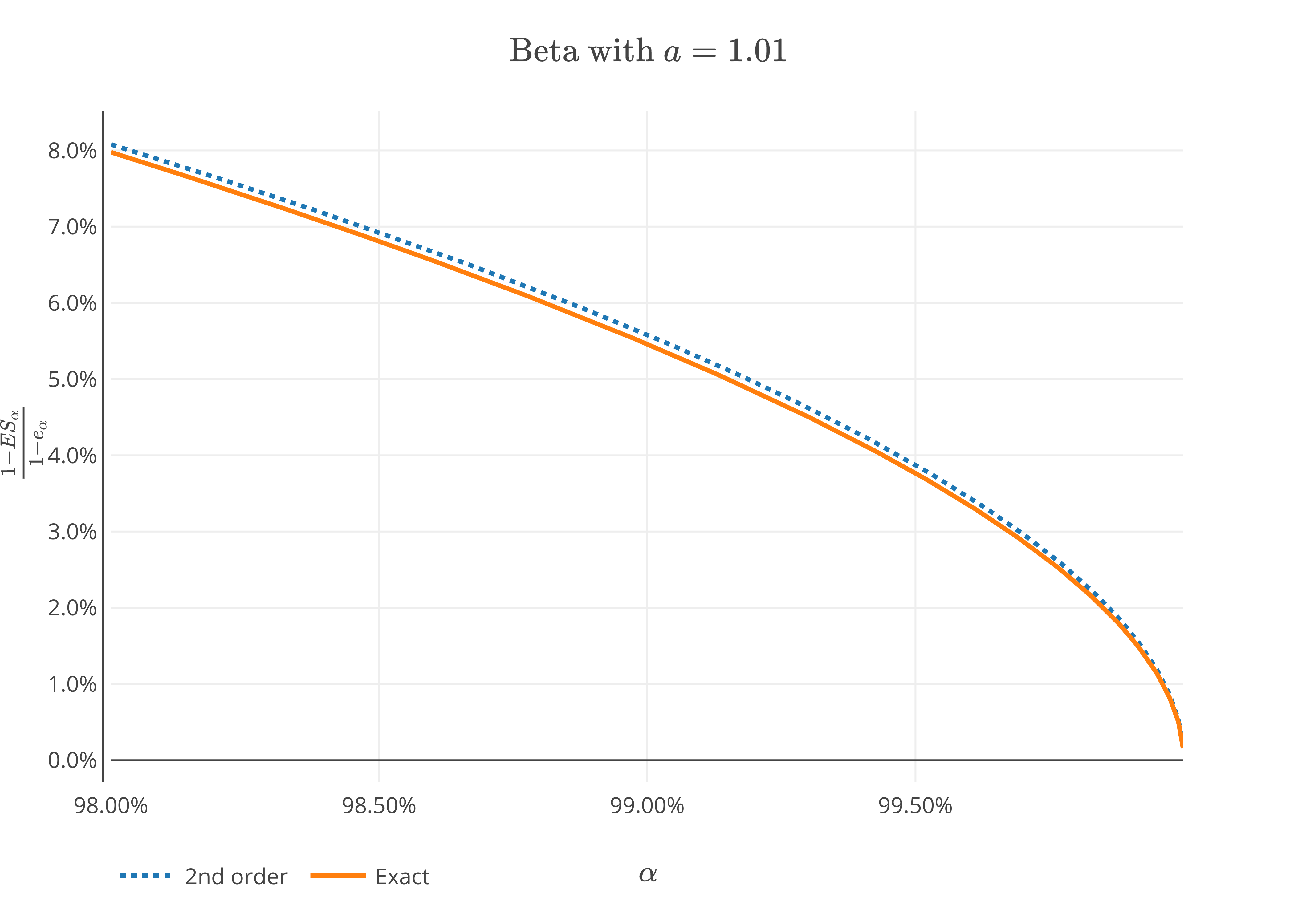}
    \end{subfigure}
    \begin{subfigure}{.49
        \textwidth}
        \includegraphics[width=\textwidth]{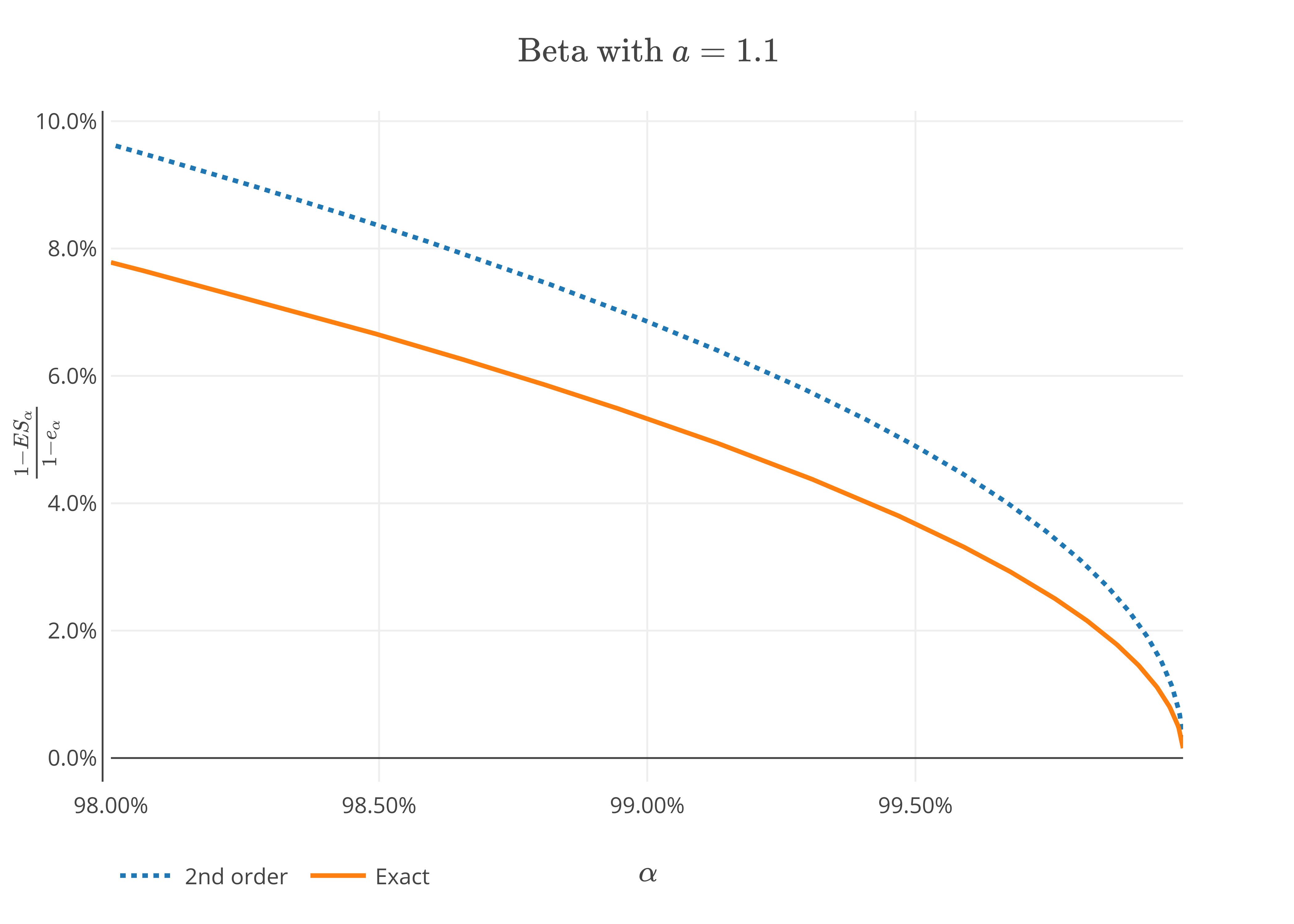}
    \end{subfigure}
    \caption{Graph of the ratio $(1-e_\alpha) /(1-ES_\alpha)$ for Beta distribution  with $a=1.01$ and $a=1.1$.}
    \label{fig:figbeta}
\end{figure}
As $\alpha$ goes to $1$, the ratio $(1-e_\alpha) /(1-ES_\alpha)$ goes to $0$.
As shown in figure \ref{fig:figbeta}, the accuracy of the second order expansion for $(1-e_\alpha) /(1-ES_\alpha)$ depends on the parameter $a$.
As $a$ close to $1$, the second order expansion become more accurate.

\begin{example}[Exponential]\label{eg:exponential}
    Let $F_L(x)=1-\exp(-x)$ for $x\geq 0$. 
    Then $ES_{\beta^\ast}(L)=1-\ln(1-\beta^\ast)$ and $q_L(\beta^\ast)=-ln(1-\beta^\ast)$.
    The Relation \eqref{eq:opt} becomes $1+(1-2\alpha)\beta^\ast=(1-\alpha)(1-\ln(1-\beta^\ast))$.
    For $x:=1-\ln{\beta^\ast}$, it holds $(x-2)e^{x-2}=(2\alpha-1)/((1-\alpha) e)$.
    Thus, $x=2+\mathcal{W}((2\alpha-1)/(1-\alpha)e)$ and 
    $\beta^\ast=1-\exp(1-x)$ where, $\mathcal{W}$ is Lambert function\footnote{$\mathcal{W}$ is a function such that $xe^x=y$ if and only if $x=\mathcal{W}(y)$.}.
    Therefore,
    \begin{equation*}
        e_\alpha(L)=1+\mathcal{W}\left(\frac{2\alpha-1}{(1-\alpha) e}\right).
    \end{equation*}
    A similar expression for $e_\alpha$ can also be found in \citep{bellini2017}.
    It is also known that $F_L$ belongs to Gumbel type $MDA(\Lambda)$ and satisfy condition \eqref{eq:GU01}.
    Hence, $e_\alpha(L)\sim ES_\alpha(L)$.
\end{example}

\begin{example}[Pareto]
    \label{eg:Pareto}
    For $a>1$ and $x\geq 0$, let $ F_L(x)=1-(1/(x+1))^a$.
    It follows that $q_L(\beta^\ast)={(1-\beta^\ast)}^{-1/a}-1$, $E[L]=1/(a-1)$ and $ES_{\beta^\ast}(L)= a E[L]{(1-\beta^\ast)}^{-1/a}-1$.
    Relation \eqref{eq:opt} gives the optimal $\beta^\ast$ solving 
    \begin{equation*}
        a(1-\alpha) \left({(1-\beta^\ast)}^{\frac{1}{a}} -1 \right)+\alpha +(1-2\alpha)\beta^\ast=0.
    \end{equation*}
    Hence,
    \begin{equation*}
        e_\alpha(L)={(1-\beta^\ast)}^{-1/a}-1.
    \end{equation*} 
    In particular, for $a=2$,
    \begin{equation*}
        \beta^\ast=\frac{\alpha+2\sqrt{\alpha(1-\alpha)}}{1+2\sqrt{\alpha(1-\alpha)}}, \quad \text{and} \quad e_\alpha(L)= \frac{\sqrt{\alpha(1-\alpha)}}{1-\alpha}.
    \end{equation*}
    It also holds that $1-F_L\in 2RV_{-a,-1}$ with auxiliary function $A(x)=a/x$, see \citep{hua2011,mao2015}.
    By Remark \ref{rem:FR}, for $\tilde{L}=L-E[L]$ it follows that $1-F_{\tilde{L}}\in 2RV_{-a,-1}$ with auxiliary function $A^\ast(x)=a^2x^{-1}/(a-1)$.
    Hence, by Proposition \ref{prop:frechet2}, it holds that
    \begin{equation*}
        \frac{e_\alpha(\tilde{L})}{ES_\alpha(\tilde{L})}=
        \frac{(a-1)^{\frac{a-1}{a}}(2\alpha-1)^{\frac{1}{a}}}{a}\left(1+\frac{1-(a-1)^{\frac{1}{a}}}{(1-\alpha)^{-\frac{1}{a}}-\frac{a}{a-1}}(1+o(1))\right).
    \end{equation*}
    The cash-invariant property gives 
    \begin{equation*}
        e_\alpha(L)=\frac{1}{a-1}+\left(\frac{2\alpha-1}{a-1}\right)^{\frac{1}{a}}\left((1-\alpha)^{-\frac{1}{a}}-1\right)\left(1+\frac{1-(a-1)^{\frac{1}{a}}}{(1-\alpha)^{-\frac{1}{a}}-\frac{a}{a-1}}(1+o(1))\right).
    \end{equation*}
    \begin{figure}[H]
        \centering
        \begin{subfigure}{.49\textwidth}
            \includegraphics[width=\textwidth]{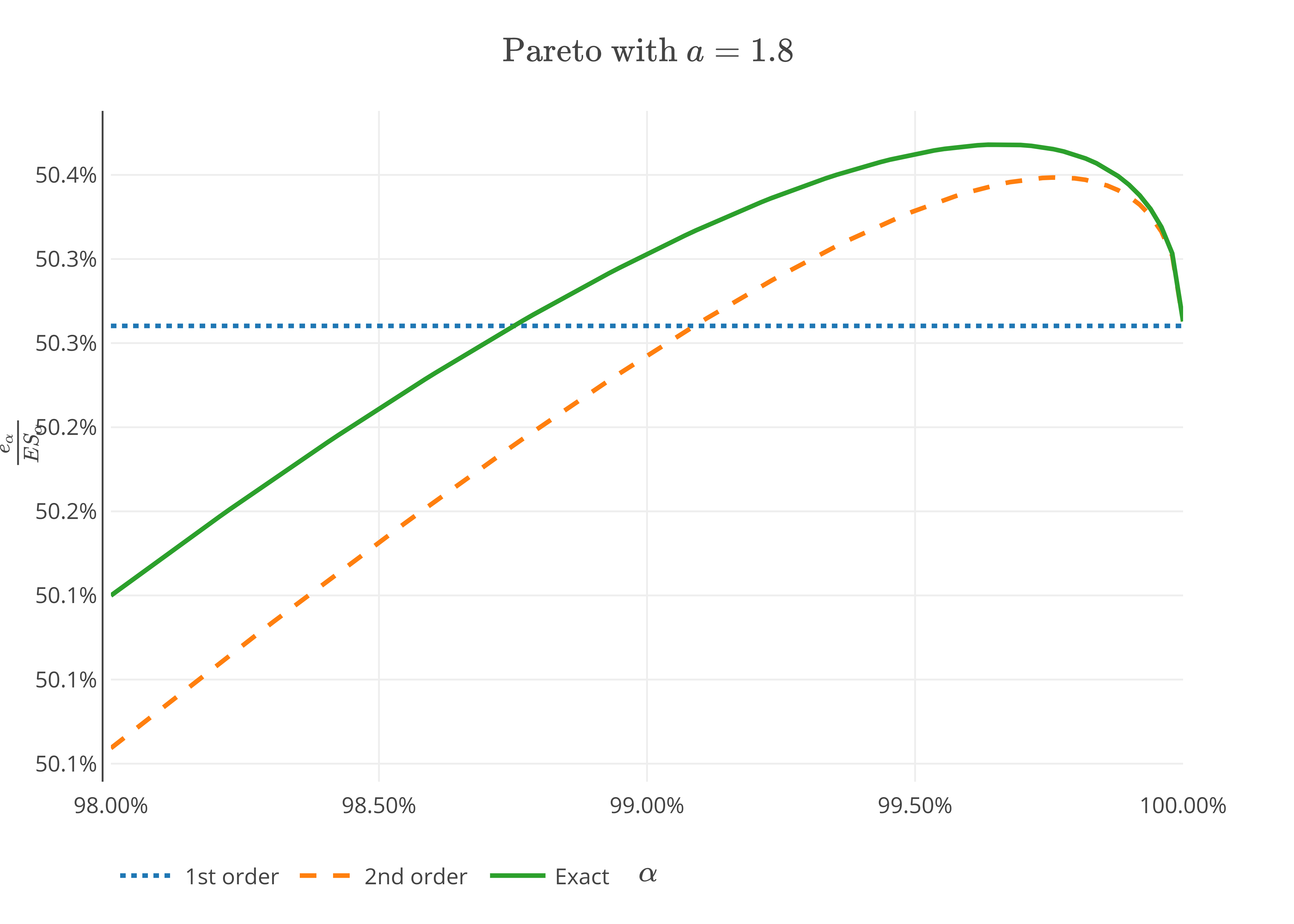}
        \end{subfigure}
        \begin{subfigure}{.49\textwidth}
            \includegraphics[width=\textwidth]{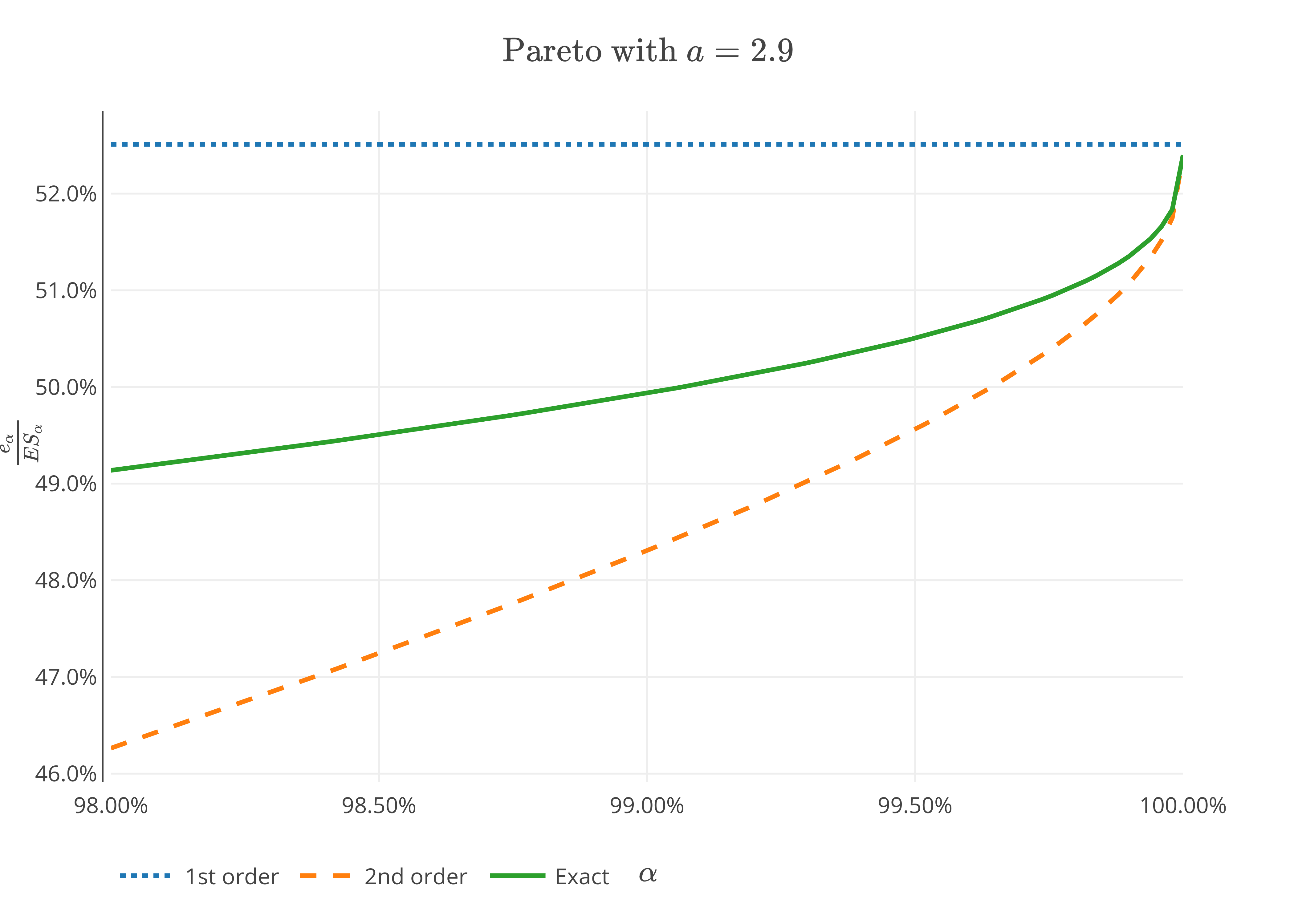}
        \end{subfigure}
        \caption{Graph of $e_\alpha(\tilde{L}) /ES_\alpha(\tilde{L})$ for Pareto distribution for $a=1.8$ and $a=2.9$.}
        \label{fig:Pareto01}
    \end{figure}
    The second order expansion is more accurate than the first order one.
    The accuracy seems better when the tail become more heavier, see Figure \ref{fig:Pareto01}.
\end{example}

\begin{example}[Standard Student t]\label{eg:Student}
    Let $L$ be a standard Student $t$ with degree of freedom $v>1$. 
    From \citep{neil2015}, we get 
    \begin{equation*}
        ES_{\beta^\ast}(L)=\frac{1}{(1-\beta^\ast)(v-1)}\psi(\Psi^{-1}(\beta^\ast)) \left(v+(\Psi^{-1}(\beta^\ast))^2\right)
    \end{equation*}
    where $\Psi$ and $\psi$ are the cumulative distribution and probability density function of the standard Student $t$ distribution with $v$ degree of freedom, respectively.
    Relation \eqref{eq:opt}, yields the optimal $\beta^\ast$ solving 
    \begin{equation*}
        \Psi^{-1}(\beta^\ast)=\frac{(2\alpha-1)\psi(\Psi^{-1}(\beta^\ast))(v+(\Psi^{-1}(\beta^\ast))^2)}{(v-1)((1-2\alpha)\beta^\ast+\alpha)}.
    \end{equation*}
    Hence, 
    \begin{equation*}
        e_\alpha(L)=\frac{(2\alpha-1)\psi(\Psi^{-1}(\beta^\ast))(v+(\Psi^{-1}(\beta^\ast))^2)}{(v-1)((1-2\alpha)\beta^\ast+\alpha)}.
    \end{equation*}
    It also holds that $\eta=\nu$ such that $1-F_L\in 2RV_{-\nu,-2}$ with auxiliary function $A(x)=\nu^2(\nu+1)x^{-2}/(\nu+2)$, see \citep{hua2011,mao2015}.
    By Proposition \ref{prop:frechet2}, it holds that
    \begin{equation*}
        \frac{e_\alpha(L)}{ES_\alpha(L)}=
        \frac{(\nu-1)^{\frac{\nu-1}{\nu}}(2\alpha-1)^{\frac{1}{\nu}}}{\nu}\left(1+\frac{(\nu-1)\left(1-(\nu-1)^{\frac{2}{\nu}}\right)}{2(\nu+2)(q_L^2(\alpha))}(1+o(1))\right).
    \end{equation*}
    Figure \ref{fig:Studentt01} shows that the second order expansion is more accurate than the first order one.
    The simulation results suggests that the accuracy of the second order expansion is better when the distribution become more heavier.
    \begin{figure}[H]
        \centering
        \begin{subfigure}{.49\textwidth}
            \includegraphics[width=\textwidth]{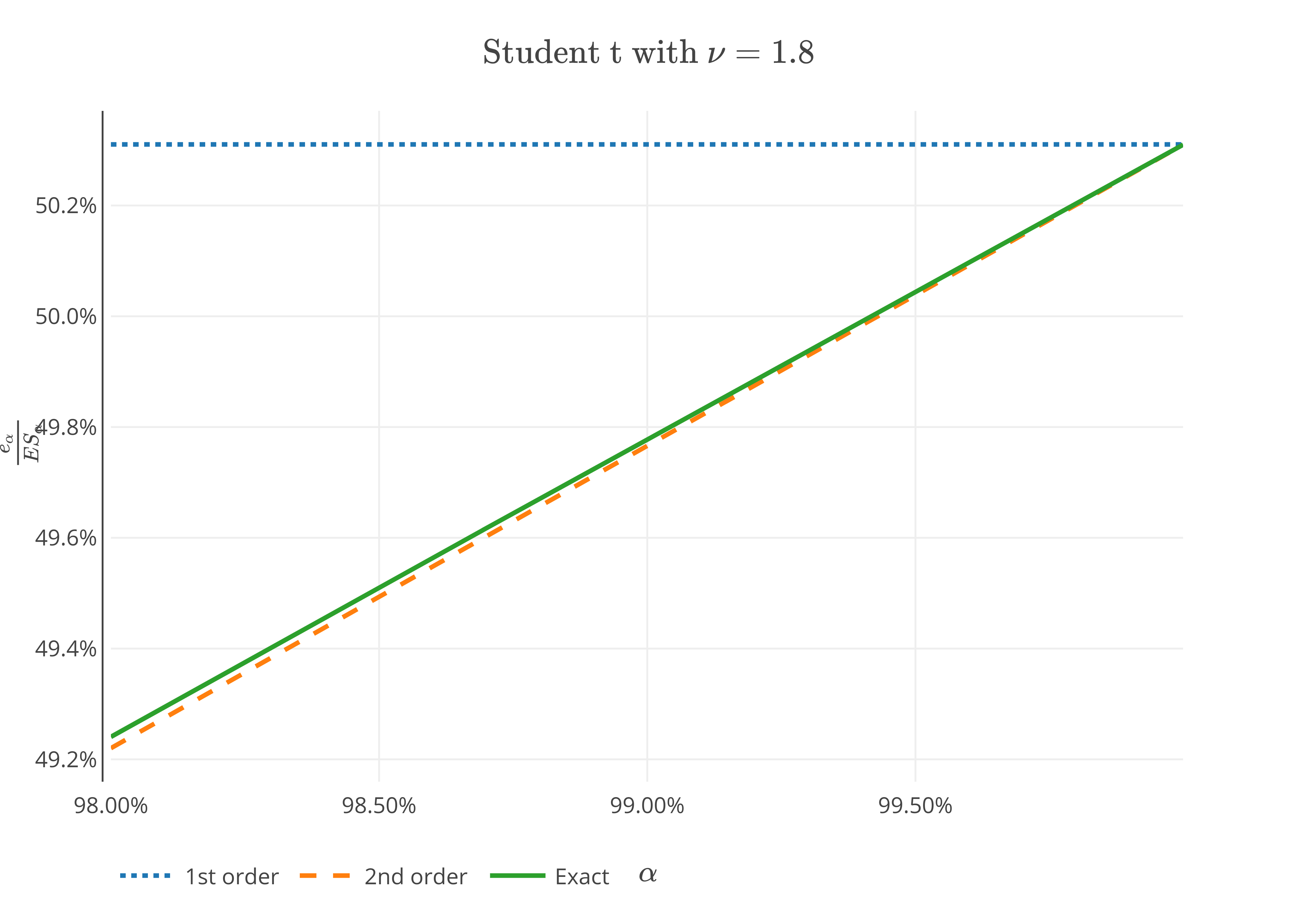}
        \end{subfigure}
        \begin{subfigure}{.49\textwidth}
            \includegraphics[width=\textwidth]{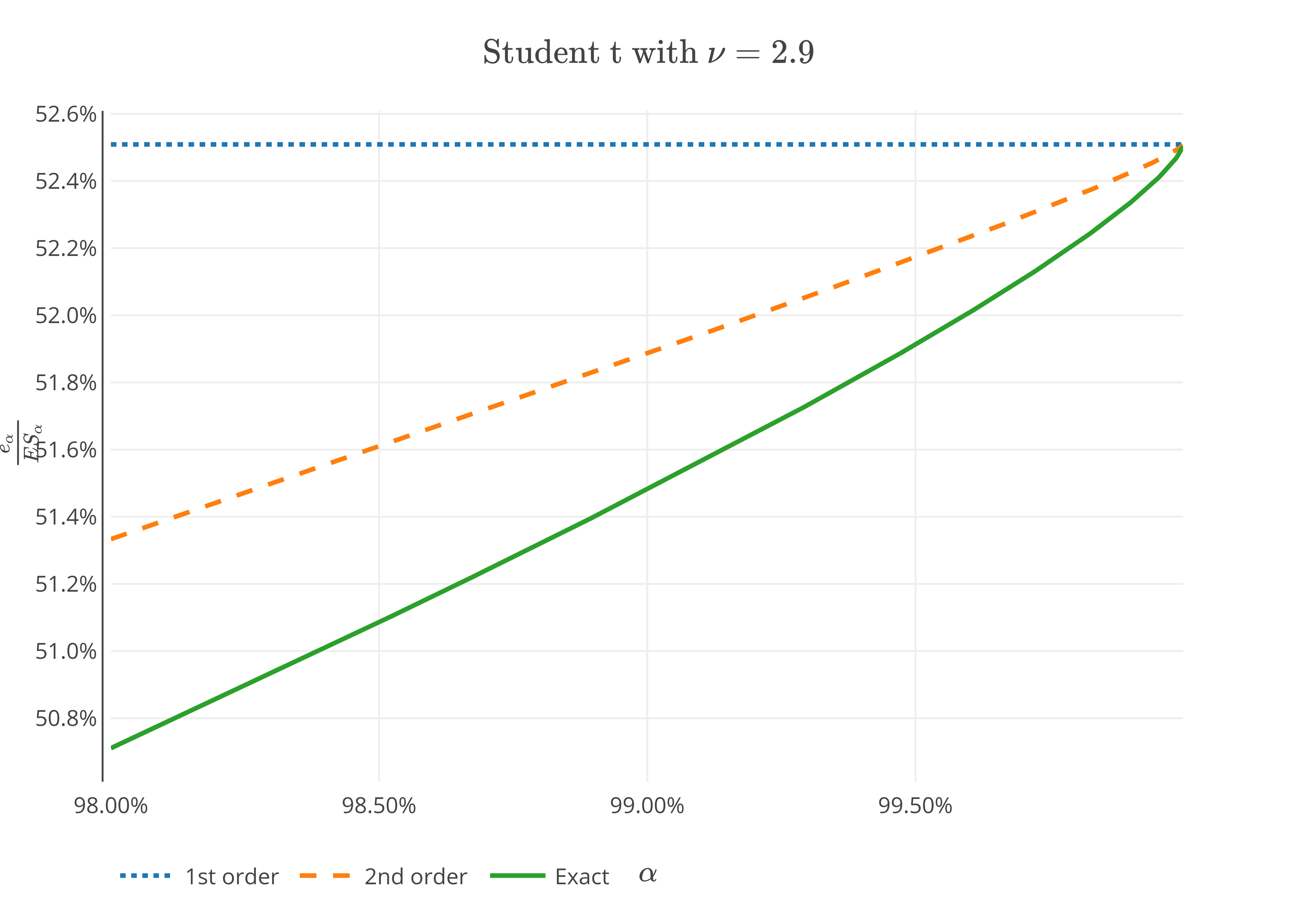}
        \end{subfigure}
        \caption{Graph of $e_\alpha /ES_\alpha$ for standard Student $t$  distribution for $\nu=1.8$ and $\nu=2.9$.}
        \label{fig:Studentt01}
    \end{figure}
\end{example}
As shown in Section \ref{sec:04}, in order to have the same bound for the probability of an estimation error bigger than a fixed threshold $\varepsilon$, value at risk  needs more observations than  expected shortfall and expectile as $\alpha$ goes to $1$, when the data is sampled from a heavy tailed distribution.
As a result of this fact, the same argument can be applied for the ratio of expectile to quantile versus the ratio of expectile to expected shortfall.
To illustrate this fact, we compare the empirical ratios $e_{\alpha,n}/q_{\alpha,n}$ and $e_{\alpha,n}/ES_{\alpha,n}$ with the theoretical ratio $e_\alpha/q_\alpha$ and $e_\alpha/ES_\alpha$, respectively by generating a random sample from Pareto and Standard Student $t$ distributions with tail index $\eta=2.1$ and $\eta=2.3$.
We denote by $Err\%$ the absolute relative percentage error 
\begin{equation*}
    \left|\frac{e_\alpha/q_\alpha-e_{\alpha,n}/q_{\alpha,n}}{e_\alpha/q_\alpha}\right|\times 100\% \quad \text{and}\quad
    \left|\frac{e_\alpha/ES_\alpha-e_{\alpha,n}/ES_{\alpha,n}}{e_\alpha/ES_\alpha}\right|\times 100\%
\end{equation*}
of the empirical ratio to the theoretical ratio for quantile and expected shortfall, respectively.
Table \ref{tab:pareto21}--\ref{tab:T17} compares the relative percentage error of the empirical ratios $e_{\alpha,n}/q_{\alpha,n}$ and $e_{\alpha,n}/ES_{\alpha,n}$ with the theoretical ratio $e_\alpha/q_\alpha$ and $e_\alpha/ES_\alpha$, respectively.
Indeed, both tables suggest that for Pareto and Student $t$ distributions the ratio $e_{\alpha,n}/ES_{\alpha,n}$ converges to  $e_\alpha/ES_\alpha$ faster than the ratio of expectile to value at risk.
\begin{table}[H]
    \begin{center}
        \resizebox{\columnwidth}{!}{
            \begin{tabular}{@{}lrrrrrrrrrrr@{}}
                \toprule
                \multicolumn{1}{c}{} & \multicolumn{1}{c}{} & \multicolumn{1}{c}{} & \multicolumn{1}{c}{}  & \multicolumn{2}{c}{$n=10^6$} & \multicolumn{1}{c}{}& \multicolumn{2}{c}{$n=5\times 10^5$}& \multicolumn{1}{c}{}& \multicolumn{2}{c}{$n=10^5$}\\
                    \cmidrule{5-6} \cmidrule{8-9} \cmidrule{11-12}
                    {}& $\alpha$& Theo. Ratio &{}& Emp. Ratio& $Err\%$ &{} & Emp. Ratio& $Err\%$ &{}& Emp. Ratio & $Err\%$\\
                    \midrule
                               &$98.3\%$  &$1.0941$   &{}        &$1.1066$
                               &$1.14\%$  &{}         &$1.0948$  &$0.05\%$
                               &{}        &$1.0697$   &$2.24\%$\\
                Expectile vs   &$98.7\%$  &$1.0759$   &{}        &$1.0922$ 
                               &$1.53\%$  &{}         &$1.0753$  &$0.06\%$
                               &{}        &$1.0593$   &$1.55\%$\\
                Value at Risk  &$99.1\%$  &$1.0551$   &{}        &$1.0742$
                               &$1.81\%$  &{}         &$1.0557$  &$0.05\%$
                               &{}        &$1.0499$   &$0.49\%$\\
                               &$99.5\%$  &$1.0294$   &{}        &$1.0491$
                               &$1.91\%$  &{}         &$1.0373$  &$0.76\%$
                               &{}        &$0.9940$   &$3.44\%$\\
                               &$99.9\%$  &$0.9888$   &{}         &$1.0238$
                               &$3.53\%$  &{}         &$0.9805$   &$0.84\%$
                               &{}        &$0.9456$   &$4.37\%$\\
                \midrule
                               &$98.3\%$  &$0.5307$   &{}         &$0.5308$
                               &$0.02\%$  &{}         &$0.5307$   &$0.00\%$
                               &{}        &$0.5310$   &$0.05\%$\\
                               &$98.7\%$  &$0.5273$   &{}         &$0.5275$
                               &$0.04\%$  &{}         &$0.5272$   &$0.02\%$
                               &{}        &$0.5278$   &$0.11\%$\\
                Expectile vs   &$99.1\%$  &$0.5231$   &{}         &$0.5223$
                               &$0.03\%$  &{}         &$0.5231$   &$4.50\%$
                               &{}        &$0.5238$   &$0.13\%$\\
                ES             &$99.5\%$  &$0.5177$   &{}         &$0.5177$
                               &$0.00\%$  &{}         &$0.5177$   &$0.01\%$
                               &{}        &$0.5183$   &$0.11\%$\\
                               &$99.9\%$  &$0.5086$   &{}         &$0.5083$
                               &$0.05\%$  &{}         &$0.5086$   &$0.00\%$
                               &{}        &$0.5103$   &$0.34\%$\\
                \bottomrule
            \end{tabular}
        }
    \end{center}
    \caption{The ratio $e_{\alpha,n}/q_{\alpha,n}$,  $e_\alpha /q_\alpha$, $e_{\alpha,n}/ES_{\alpha,n}$, $e_\alpha /ES_\alpha$ and relative percentage error of Pareto distribution with $a=2.1$.}
    \label{tab:pareto21}
\end{table}
\begin{table}[H]
    \begin{center}
        \resizebox{\columnwidth}{!}{
            \begin{tabular}{@{}lrrrrrrrrrrr@{}}
                \toprule
                \multicolumn{1}{c}{} & \multicolumn{1}{c}{} & \multicolumn{1}{c}{} & \multicolumn{1}{c}{}  & \multicolumn{2}{c}{$n=10^6$} & \multicolumn{1}{c}{}& \multicolumn{2}{c}{$n=5\times 10^5$}& \multicolumn{1}{c}{}& \multicolumn{2}{c}{$n=10^5$}\\
                    \cmidrule{5-6} \cmidrule{8-9} \cmidrule{11-12}
                    {}& $\alpha$& Theo. Ratio &{}& Emp. Ratio& $Err\%$ &{} & Emp. Ratio& $Err\%$ &{}& Emp. Ratio & $Err\%$\\
                    \midrule
                               &$98.3\%$  &$0.9577$   &{}        &$0.9422$
                               &$4.71\%$  &{}         &$0.9387$  &$4.34\%$
                               &{}        &$0.9839$   &$9.07\%$\\
                Expectile vs   &$98.7\%$  &$0.9574$   &{}        &$0.9386$ 
                               &$4.42\%$  &{}         &$0.9363$  &$4.17\%$
                               &{}        &$0.9892$   &$9.69\%$\\
                Value at Risk  &$99.1\%$  &$0.9569$   &{}        &$0.9370$
                               &$4.34\%$  &{}         &$0.9268$  &$3.27\%$
                               &{}        &$0.9929$   &$10.18\%$\\
                               &$99.5\%$  &$0.9565$   &{}        &$0.9269$
                               &$3.39\%$  &{}         &$0.9220$  &$2.88\%$
                               &{}        &$0.9913$   &$10.13\%$\\
                               &$99.9\%$  &$0.9559$   &{}         &$0.9180$
                               &$2.63\%$  &{}         &$0.8859$   &$0.73\%$
                               &{}        &$0.9593$   &$6.94\%$\\
                \midrule
                               &$98.3\%$  &$0.4918$   &{}         &$0.4919$
                               &$0.56\%$  &{}         &$0.4919$   &$0.56\%$
                               &{}        &$0.4911$   &$0.73\%$\\
                               &$98.7\%$  &$0.4938$   &{}         &$0.4942$
                               &$0.56\%$  &{}         &$0.4942$   &$0.55\%$
                               &{}        &$0.4932$   &$0.77\%$\\
                Expectile vs   &$99.1\%$  &$0.4959$   &{}         &$0.4964$
                               &$0.56\%$  &{}         &$0.4965$   &$0.52\%$
                               &{}        &$0.4952$   &$0.79\%$\\
                ES             &$99.5\%$  &$0.4980$   &{}         &$0.4989$
                               &$0.48\%$  &{}         &$0.4990$   &$0.47\%$
                               &{}        &$0.4973$   &$0.81\%$\\
                               &$99.9\%$  &$0.5000$   &{}         &$0.5017$
                               &$0.38\%$  &{}         &$0.5030$   &$0.12\%$
                               &{}        &$0.4998$   &$0.78\%$\\
                \bottomrule
            \end{tabular}
        }
    \end{center}
    \caption{The ratio $e_{\alpha,n}/q_{\alpha,n}$,  $e_\alpha /q_\alpha$, $e_{\alpha,n}/ES_{\alpha,n}$, $e_\alpha /ES_\alpha$ and relative percentage error for standard Student $t$ distribution with $\nu=2.1$.}
    \label{tab:T21}
\end{table}

The simulation result also suggest that with $\nu=2.1$ for  Student $t$ distribution  more observations may be needed for  $e_{\alpha,n}/q_{\alpha,n}$ than  $e_{\alpha,n}/ES_{\alpha,n}$ to converge to the theoretical ratio as compared to Pareto distribution with $a=2.1$, see Table \ref{tab:pareto21} and \ref{tab:T21}.
\begin{table}[H]
    \begin{center}
        \resizebox{\columnwidth}{!}{
            \begin{tabular}{@{}lrrrrrrrrrrr@{}}
                \toprule
                \multicolumn{1}{c}{} & \multicolumn{1}{c}{} & \multicolumn{1}{c}{} & \multicolumn{1}{c}{}  & \multicolumn{2}{c}{$n=10^6$} & \multicolumn{1}{c}{}& \multicolumn{2}{c}{$n=5\times 10^5$}& \multicolumn{1}{c}{}& \multicolumn{2}{c}{$n=10^5$}\\
                    \cmidrule{5-6} \cmidrule{8-9} \cmidrule{11-12}
                    {}& $\alpha$& Theo. Ratio &{}& Emp. Ratio& $Err\%$ &{} & Emp. Ratio& $Err\%$ &{}& Emp. Ratio & $Err\%$\\
                    \midrule
                               &$98.3\%$  &$1.0272$   &{}        &$1.0364$
                               &$0.88\%$  &{}         &$1.0136$  &$1.34\%$
                               &{}        &$1.0173$   &$0.97\%$\\
                Expectile vs   &$98.7\%$  &$1.0102$   &{}        &$1.0229$ 
                               &$1.26\%$  &{}         &$0.9992$  &$1.08\%$
                               &{}        &$0.9921$   &$1.79\%$\\
                Value at Risk  &$99.1\%$  &$0.9905$   &{}        &$1.0051$
                               &$1.47\%$  &{}         &$0.9785$  &$1.21\%$
                               &{}        &$0.9576$   &$3.33\%$\\
                               &$99.5\%$  &$0.9661$   &{}        &$0.9794$
                               &$1.37\%$  &{}         &$0.9474$  &$1.94\%$
                               &{}        &$0.9229$   &$4.47\%$\\
                               &$99.9\%$  &$0.9270$   &{}         &$0.9630$
                               &$3.88\%$  &{}         &$0.8820$   &$4.85\%$
                               &{}        &$0.8852$   &$4.52\%$\\
                \midrule
                               &$98.3\%$  &$0.5331$   &{}         &$0.5330$
                               &$0.02\%$  &{}         &$0.5335$   &$0.07\%$
                               &{}        &$0.5340$   &$0.17\%$\\
                               &$98.7\%$  &$0.5299$   &{}         &$0.5297$
                               &$0.03\%$  &{}         &$0.5304$   &$0.09\%$
                               &{}        &$0.5309$   &$0.19\%$\\
                Expectile vs   &$99.1\%$  &$0.5260$   &{}         &$0.5257$
                               &$0.06\%$  &{}         &$0.5266$   &$0.12\%$
                               &{}        &$0.5275$   &$0.29\%$\\
                ES             &$99.5\%$  &$0.5209$   &{}         &$0.5204$
                               &$0.11\%$  &{}         &$0.5218$   &$0.17\%$
                               &{}        &$0.5238$   &$0.55\%$\\
                               &$99.9\%$  &$0.5123$   &{}         &$0.5107$
                               &$0.33\%$  &{}         &$0.5157$   &$0.65\%$
                               &{}        &$0.5184$   &$1.18\%$\\
                \bottomrule
            \end{tabular}
        }
    \end{center}
    \caption{The Ratio $e_{\alpha,n}/q_{\alpha,n}$,  $e_\alpha/q_\alpha$, $e_{\alpha,n}/ES_{\alpha,n}, $ $e_\alpha/ES_\alpha$ and relative percentage error  for Pareto distribution with $a=2.3$.}
        \label{tab:pareto17}
\end{table}
\begin{table}[H]
    \begin{center}
        \resizebox{\columnwidth}{!}{
            \begin{tabular}{@{}lrrrrrrrrrrr@{}}
                \toprule
                \multicolumn{1}{c}{} & \multicolumn{1}{c}{} & \multicolumn{1}{c}{} & \multicolumn{1}{c}{}  & \multicolumn{2}{c}{$n=10^6$} & \multicolumn{1}{c}{}& \multicolumn{2}{c}{$n=5\times 10^5$}& \multicolumn{1}{c}{}& \multicolumn{2}{c}{$n=10^5$}\\
                    \cmidrule{5-6} \cmidrule{8-9} \cmidrule{11-12}
                    {}& $\alpha$& Theo. Ratio &{}& Emp. Ratio& $Err\%$ &{} & Emp. Ratio& $Err\%$ &{}& Emp. Ratio & $Err\%$\\
                    \midrule
                               &$98.3\%$  &$0.8971$   &{}        &$0.8967$
                               &$0.04\%$  &{}         &$0.8962$  &$0.09\%$
                               &{}        &$0.9190$   &$2.45\%$\\
                Expectile vs   &$98.7\%$  &$0.8963$   &{}        &$0.8953$ 
                               &$0.12\%$  &{}         &$0.8940$  &$0.26\%$
                               &{}        &$0.9165$   &$2.25\%$\\
                Value at Risk  &$99.1\%$  &$0.8955$   &{}        &$0.8925$
                               &$0.34\%$  &{}         &$0.9034$  &$0.88\%$
                               &{}        &$0.9211$   &$2.86\%$\\
                               &$99.5\%$  &$0.8944$   &{}        &$0.8976$
                               &$0.36\%$  &{}         &$0.8990$  &$0.51\%$
                               &{}        &$0.9130$   &$2.07\%$\\
                               &$99.9\%$  &$0.8929$   &{}         &$0.8933$
                               &$0.04\%$  &{}         &$0.8666$   &$2.94\%$
                               &{}        &$0.9429$   &$5.60\%$\\
                \midrule
                               &$98.3\%$  &$0.4947$   &{}         &$0.4946$
                               &$0.04\%$  &{}         &$0.4950$   &$0.05\%$
                               &{}        &$0.4932$   &$0.32\%$\\
                               &$98.7\%$  &$0.4969$   &{}         &$0.4970$
                               &$0.00\%$  &{}         &$0.4972$   &$0.04\%$
                               &{}        &$0.4951$   &$0.37\%$\\
                Expectile vs   &$99.1\%$  &$0.4991$   &{}         &$0.4992$
                               &$0.02\%$  &{}         &$0.4989$   &$0.05\%$
                               &{}        &$0.4972$   &$0.40\%$\\
                ES             &$99.5\%$  &$0.5014$   &{}         &$0.5012$
                               &$0.04\%$  &{}         &$0.5008$   &$0.11\%$
                               &{}        &$0.4996$   &$0.36\%$\\
                               &$99.9\%$  &$0.5037$   &{}         &$0.5033$
                               &$0.07\%$  &{}         &$0.5066$   &$0.56\%$
                               &{}        &$0.5009$   &$0.55\%$\\
                \bottomrule
            \end{tabular}
        }
    \end{center}
    \caption{The ratio $e_{\alpha,n}/q_{\alpha,n}$,  $e_\alpha /q_\alpha$, $e_{\alpha,n}/ES_{\alpha,n}$, $e_\alpha /ES_\alpha$ and relative percentage error for standard Student $t$ distribution with $\nu=2.3$.}
        \label{tab:T17}
\end{table}

\bibliographystyle{abbrvnat}
\bibliography{biblio}
\end{document}